\documentclass[% reprint,
superscriptaddress,
frontmatterverbose, 
 amsmath,amssymb,
 aps,
 pra,
notitlepage
]{revtex4-1}
\usepackage[english]{babel}
\usepackage{amsmath, amssymb, amsfonts, mathrsfs}
\usepackage{ntheorem}
\usepackage{subcaption}
\usepackage{graphicx}% Include figure files
\usepackage{dcolumn}% Align table columns on decimal point
\usepackage{bm}% bold math
\usepackage[shortlabels]{enumitem} 
\usepackage{dsfont}
\usepackage{chngcntr}
\usepackage{apptools}
\usepackage{mathtools}
\AtAppendix{\counterwithin{cor}{section}}
\AtAppendix{\counterwithin{theorem}{section}}
\usepackage{xcolor,soul}
\usepackage{url} 
\usepackage{multirow}

\DeclareMathOperator{\Tr}{Tr}

\newtheorem{theorem}{Theorem}
\newtheorem{lemma}{Lemma}

\newtheorem{prop}{Proposition}
\newtheorem{cor}{Corollary}

\newtheorem{definition}{Definition}
\newtheorem{claim}{Claim}

\newtheorem{lemma2}{Lemma}[section]

\newtheorem{definition2}{Definition}[section]
\newtheorem{remark}{Remark}
\theoremstyle{nonumberbreak}
\newtheorem{proof}{Proof}

\usepackage[colorlinks = true,
            urlcolor=blue
            ]{hyperref}

\usepackage[capitalise]{cleveref}
\crefname{cor}{Corollary}{Corollaries}
\crefname{obs}{Observation}{Observations}
\crefname{remark}{Remark}{Remarks}
\usepackage{epstopdf}
\usepackage{tikz}
\begin{document}

%\preprint{APS/123-QED}

\title{Dimensionally sharp inequalities for the linear entropy}
\author{Simon Morelli}
\email{simon.morelli@hotmail.com}
\affiliation{
Institute for Quantum Optics and Quantum Information - IQOQI Vienna, Austrian Academy of Sciences, Boltzmanngasse 3, 1090 Vienna, Austria}
\author{Claude Kl\"ockl}%
 \email{claudio.kloeckl@reflex.at}
 \affiliation{
Institute for Quantum Optics and Quantum Information - IQOQI Vienna, Austrian Academy of Sciences, Boltzmanngasse 3, 1090 Vienna, Austria}
\affiliation{
Institute for sustainable economic development, University of Natural Resources and Life Sciences (BOKU), Gregor-Mendel-Stra\ss e 33,
1180 Vienna, Austria }%
\author{Christopher Eltschka}%
\email{christopher.eltschka@physik.uni-regensburg.de}
\affiliation{Institut f\"ur Theoretische Physik, Universit\"at Regensburg, D-93040 Regensburg, Germany}
\author{Jens Siewert}%
\email{jens.siewert@ehu.eus}
\affiliation{Departamento de Quim\'ica F\'isica, Universidad del Pa\'is Vasco UPV/EHU, E-48080 Bilbao, Spain}
\affiliation{IKERBASQUE Basque Foundation for Science, E-48013 Bilbao,Spain}
\author{Marcus Huber}%
\email{entangledanarchist@gmail.com}
\affiliation{
Institute for Quantum Optics and Quantum Information - IQOQI Vienna, Austrian Academy of Sciences, Boltzmanngasse 3, 1090 Vienna, Austria}

\date{\today}% It is always \today, today,
             %  but any date may be explicitly specified

\begin{abstract}
\noindent
We derive an inequality for the linear entropy, that gives sharp bounds for all finite dimensional systems. The derivation is based on generalised Bloch decompositions and provides a strict improvement for the possible distribution of purities for all finite dimensional quantum states. It thus extends the widely used concept of entropy inequalities from the asymptotic to the finite regime, and should also find applications in entanglement detection and efficient experimental characterisations of quantum states.
\end{abstract}

\pacs{Valid PACS appear here}

\maketitle

\section{Introduction}

Entropy is a widely used concept in various fields. In information theory, entropy is a measure of uncertainty or lack of knowledge about a given system described by a random variable. In a quantum setting this random variable is a density operator, describing the state of a system, and the entropy is a function of this density operator. So the entropy is an intrinsic property of a state and not an observable.
Entropy inequalities characterise the distribution of entropy, and thus the distribution of information, in the various parts of a multipartite system. Therefore entropy inequalities are also
referred to as the laws of information theory. Let us illustrate this concept for the best understood entropy measure, the von Neumann entropy

\begin{align}
S(\rho):=-\Tr\rho\log_2(\rho).
\end{align}

Assume we have an $n$-partite quantum system $A_{[n]}:=A_1...A_n$ in the state $\rho_{A_1...A_n}\in\mathcal{H}_{A_1}\otimes\dots\otimes\mathcal{H}_{A_n}$.
The reduced density operator $\rho_{A_I}=\rho_{A_{i_1}...A_{i_{|I|}}}:=\Tr_{A_{I^c}}(\rho_{A_{[n]}})$ (where $\emptyset \neq I \subseteq [n]$) describes the state of the subsystem $A_I$.
We call the vector
\begin{equation}
\vec{\mathbf{v}}(\rho_{A_{[n]}}):=\big (S(\rho_{A_1}),...,S(\rho_{A_n}),S(\rho_{A_1A_2}),...,S(\rho_{A_1...A_n})\big)\label{VN6}
\end{equation}
an entropy vector.
Which points in $\mathbb{R}^{2^n-1}$ are entropy vectors? 
Since the density operator is non-negative definite and normalised to $\Tr{\rho}=1$, the entropy is always non-negative and all entropy vectors lie in the non-negative orthant. Every additional inequality reduces the possible set of entropy vectors. The remaining entropy vectors form a set, whose closure is a cone, called the entropy cone.
\begin{equation}
C_n:=\overline{\{\vec{\mathbf{v}}(\rho)|\ \rho \text{ a density matrix}\}}
\end{equation}
For a detailed discussion in the classical case see Chapter 13. -15. of \cite{Yeung2008}.
\par

For the von Neumann entropy the best-known inequalities are

\begin{align}
&\text{non-negativity}:\quad &S(\rho_A)\ge 0\\
&\text{triangle inequality or Araki-Lieb (AL)}:\quad &S(\rho_{AB})\ge|S(\rho_{A})-S(\rho_{B})|\\
&\text{subadditivity (SA)}:\quad &S(\rho_A)+S(\rho_B)\ge S(\rho_{AB})\\
&\text{weak monotonicity (WM)}:\quad &S(\rho_{AB})+S(\rho_{BC})\ge S(\rho_A)+S(\rho_C)\label{VN4}\\
&\text{strong subadditivity (SSA)}:\quad &S(\rho_{AB})+S(\rho_{BC})\ge S(\rho_{ABC})+S(\rho_B)\label{VN5}
\end{align}

This set of inequalities is not minimal, in fact they can all be derived from strong subadditivity (\ref{VN5}), or equivalently weak monotonicity (\ref{VN4}), see \cite{Nielsen2010} and \cite{Cadney12}. An inequality derived from this set (i.e. SSA) is called a von Neumann inequality. It turns out, that every inequality for bi- or tripartite quantum systems can by derived from strong subadditivity. For the classical analogue, the Shannon entropy, there are known counterexamples for n-partite systems with $n\ge4$. Contrary to the Shannon entropy, no unconstrained non-von Neumann inequalities are known for n-partite systems $(n\ge 4)$ yet. See \cite{Pippenger2003}, \cite{Linden2005} and \cite{Cadney12}.
\par

If the axioms defining the von Neumann entropy uniquely (\cite{Klir2006}) are weakened, they allow more general classes of entropies. The most prominent families of parametrised entropies are the R\'enyi entropy and the Tsallis entropy.

\begin{definition}
	For a density operator $\rho\in\mathcal{H}$ and a positive constant $\alpha$, the R\'{e}nyi entropy $S^\alpha(\rho)$ is given by 
	\begin{align}
	S^\alpha(\rho):=\frac{1}{1-\alpha}\log_2 \Tr(\rho^{\alpha})
	\end{align}
\end{definition}

It has been shown in \cite{Linden2013}, that the only possible inequalities for every n-partite system for the R\'{e}nyi entropy with $0<\alpha<1$ are the non-negativity of the entropy of the system and all possible subsystems.\\
\textit{"Somewhat surprisingly, we find for $0 < \alpha < 1$, that the
	only inequality is non-negativity: In other words, any collection of non-negative numbers
	assigned to the nonempty subsets of $n$ parties can be arbitrarily well approximated by the
	$\alpha$-entropies of the $2^{n-1}$ marginals of a quantum state."} \cite{Linden2013} \\
For $\alpha>1$ it is shown in the same paper that there are no further linear (even homogeneous) inequalities. But since other inequalities are known, the set of entropy vectors is not a cone anymore for $\alpha>1$.\\
For the special case $\alpha=0$, the R\'enyi entropy becomes the quantum Hartley measure, also known as the Schmidt rank of the density operator. For this entropy there are again other inequalities besides non-negativity, see \cite{Cadney2014}.

\begin{definition}
	The Quantum Tsallis q-entropy is defined as (see \cite{Raggio1995})
	\begin{align}
	S_q(\rho):=\frac{1}{q-1}(1-\Tr(\rho^q))
	\end{align}
	for $\rho$ a density matrix and $q>0$ a real number.
\end{definition}

For $q>1$ the Tsallis entropy is bounded by $S_q(\rho)\le\frac{1-d^{1-q}}{q-1}\le\frac{1}{q-1}$, where $d$ is the dimension of the Hilbert space. Equality in the first inequality holds only for the completely mixed state.
\par

Audenaert \cite{Audenaert2007} proved 2007 that the Tsallis entropy for $q>1$ is subadditive. The subadditivity of the Tsallis entropy follows from the following lemma, but is a weaker statement.

\begin{lemma}\label{PEFl1}
	Let $\rho_{AB}$ be a density operator of a bipartite state on a finite-dimensional Hilbert space $\mathcal{H}_1\otimes\mathcal{H}_2$ and $\|\rho\|_q:=(\Tr(\rho^q))^{\frac{1}{q}}$ the Schatten $q$-norm. Then for $q>1$ the following holds
	\begin{align}
	\|\rho_{A}\|_q+\|\rho_{B}\|_q\le 1+\|\rho_{AB}\|_q\\
	\notag
	\end{align}
\end{lemma}

Further the Tsallis entropy is known to not satisfy strong subadditivity.
\par

Both the R\'{e}nyi and the Tsallis entropy, in the limit $q\rightarrow1$, converge to the von Neumann entropy. The R\'enyi entropy can be calculated from the Tsallis entropy by
\begin{align}
S^q(\rho)=\frac{1}{1-q}\log_2\big(1-(q-1)S_q(\rho)\big)\label{Tsallistorenyi}
\end{align}
so inequalities for the Tsallis entropy can be reformulated as inequalities for the R\'enyi entropy and vice versa.
\par

In this paper we restrict ourselves to  the Tsallis 2-entropy, also known as linear entropy
\begin{align}
S_L(\rho):=1-\Tr(\rho^2)
\end{align}

The linear entropy is of remarkable interest for quantum information theory. It is obviously closely related to the purity ($\gamma(\rho)=\Tr(\rho^2)=1-S_L(\rho)$) of a quantum state and is therefore also called the impurity. The purity as information measure was introduced in \cite{Fano}.
Important entanglement witnesses arise from the linear entropy. If the linear entropy of a subsystem is greater then the entropy of the composite system, it follows that there is nonclassical correlation.
Further it is easily represented in the Bloch representation. Using the Bloch decomposition (\ref{AA1t21}) and Eq. (\ref{AA1t22}) the linear entropy can be written as
\begin{align}\label{LEeq1}
S_L(\rho)=1-\frac{1}{d}\big(1+\|\vec{\mathbf{b}}\|^2\big)
\end{align}
where $\vec{\mathbf{b}}$ denotes the Bloch vector.
\par
We have already seen how the Tsallis entropy can be transformed into the R\'enyi entropy by Eq. (\ref{Tsallistorenyi}). In the case of the linear entropy this gives the R\'enyi 2-entropy, also known as collision entropy. This is an important entropy for many applications, e.g. privacy amplification in quantum cryptography.

\par
The linear entropy can be seen as the linear approximation to the von Neumann entropy at pure states. Therefore the efficient calculation in contrast to the von Neumann entropy makes it also attractive for practical use in larger systems.
\par

We already know that the linear entropy, as a special case of the quantum Tsallis entropy, is non-negative and bounded.
\begin{align}
0\le S_{L}(\rho)\le 1-\frac{1}{d}=:D
\end{align}
where we have equality in the first inequality iff $\rho$ is a pure state and equality in the second iff $\rho=\frac{\mathds{1}_d}{d}$, i.e. completely mixed.
The dimension dependent constant $D$ denotes the maximal attainable entropy for a state in a Hilbert space with dimension $d$.

Moreover, the linear entropy is subadditive and pseudo-additive for product states
\begin{align}
S_{L}(\rho_{AB})&\le S_{L}(\rho_{A})+S_{L}(\rho_{B})\\
S_{L}(\rho_{A}\otimes\rho_{B})&=S_{L}(\rho_{A})+S_{L}(\rho_{B})-S_{L}(\rho_{A})S_{L}(\rho_{B})
\end{align}

From subadditivity the Araki-Lieb inequality can be derived via purification, i.e.
\begin{align}
S_L(\rho_{AB})\ge |S_L(\rho_A)-S_L(\rho_B)|
\end{align}
For an alternative proof see \cite{Zhang2008}.

It is known that the linear entropy is not strongly subadditive, only a weaker version holds, \cite{Petz2015}.
\par

Recently Appel et al. \cite{Appel2017} have discovered two new inequalities, the first is a dimension dependent version of the strong subadditivity

\begin{theorem}\label{LEt1}
	For a tripartite quantum system $\rho_{ABC}$ we find the following entropy inequality for the linear entropy $S_{L}(\rho_{ABC})=1-\Tr(\rho_{ABC}^{2})$:
	\begin{equation}
	S_{L}(\rho_{ABC})+\frac{1}{d_{A}d_{B}}S_{L}(\rho_{C}) \le \frac{1}{d_{B}}S_{L}(\rho_{AC})+\frac{1}{d_{A}}S_{L}(\rho_{BC})+\frac{d_{A}d_{B}+1-d_{A}-d_{B}}{d_{A}d_{B}}
	\end{equation}
	\\
\end{theorem}

The second inequality is non-linear and dimension dependent 

\begin{theorem}\label{LEt2}
	For all $\rho_{AB}$ and the linear entropy or Tsallis 2-entropy
	\begin{equation}
	1-\frac{d_{A}d_{B}}{4}(1-S_{L}(\rho_{AB})+\frac{1}{d_{A}d_{B}})^{2} \le S_{L}(\rho_{A})+S_{L}(\rho_{B})-S_{L}(\rho_{A})S_{L}(\rho_{B})
	\end{equation}
	\\
\end{theorem}

Furthermore from Lemma \ref{PEFl1}. one can derive

\begin{lemma}\label{LEl1}
	For the linear entropy the following inequality holds
	\begin{align}
	S_{L}(\rho_{AB})\le& \ S_{L}(\rho_{A})+ S_{L}(\rho_{B})-2(1-\sqrt{1- S_{L}(\rho_{A})})(1-\sqrt{1- S_{L}(\rho_{B})}) \label{LEl11} 
	\end{align}
	for
	\begin{align}
	\sqrt{1- S_{L}(\rho_{A})}+\sqrt{1- S_{L}(\rho_{B})}\ge 1
	\end{align}
	and is sharper than subadditivity.
\end{lemma}

\par

\begin{proof}
	In the case  $q=2$, Lemma \ref{PEFl1}. states
	\begin{align}
	\|\rho_{A}\|_2+\|\rho_{B}\|_2\le 1+\|\rho_{AB}\|_2
	\end{align}
	Using $\|\rho\|_2=\sqrt{\Tr(\rho^2)}$ this becomes
	\begin{align}
	\sqrt{\Tr(\rho_A^2)}+\sqrt{\Tr(\rho_B^2)}-1\le \sqrt{\Tr(\rho_{AB}^2)}\label{LEl12}
	\end{align}
	Under the assumption that the left hand side is positive, both sides can be squared. Using $S_L(\rho)=1-\Tr(\rho^2)$ one gets
	\begin{align}
	S_{L}(\rho_{AB})\le& \ S_{L}(\rho_{A})+ S_{L}(\rho_{B})-2(1-\sqrt{1- S_{L}(\rho_{A})})(1-\sqrt{1- S_{L}(\rho_{B})})
	\end{align}
	But the left hand side of Equation (\ref{LEl12}) is positive exactly when
	\begin{align}
	\sqrt{1- S_{L}(\rho_{A})}+\sqrt{1- S_{L}(\rho_{B})}\ge 1
	\end{align}
	
	This inequality is sharper than subadditivity, since $0\le S_{L}(\rho_{A}),S_{L}(\rho_{B})\le1$.
\end{proof}

\section{Inhomogeneous Subadditivity}\label{IS}

The inequality in Theorem \ref{LEt1}. from \cite{Appel2017} can be reduced to give an inequality for bipartite systems.

\begin{lemma}\label{ISl1}
For any bipartite quantum system $\rho_{AB}\in\mathcal{H}_A\otimes\mathcal{H}_B$ with the two subsystems $\rho_{A}=\Tr_B(\rho_{AB})$ and $\rho_{B}=\Tr_A(\rho_{AB})$ we have
\begin{align}
S_{L}(\rho_{AB}) &\le  \frac{1}{d_{B}}S_{L}(\rho_{A})+\frac{1}{d_{A}}S_{L}(\rho_{B})+\frac{(d_{A}-1)(d_{B}-1)}{d_{A}d_{B}}
\end{align}
where $d_A=\dim(\mathcal{H}_A)$ and $d_B=\dim(\mathcal{H}_B)$ denote the dimensions of the two partitions.
\end{lemma}

Since these equations have a constant term, the inequality in Lemma \ref{ISl1}. will be called inhomogeneous subadditivity (ISA) and the inequality from Theorem \ref{LEt1}. strong inhomogeneous subadditivity (SISA).\\

This result follows either from  Theorem \ref{LEt1}. of \cite{Appel2017} and pseudo-additivity or directly from the Bloch decomposition.

\begin{proof}
If we take the inequality in Theorem \ref{LEt1}.
\begin{align}
S_{L}(\rho_{ABC})+\frac{1}{d_{A}d_{B}}S_{L}(\rho_{C}) \le \frac{1}{d_{B}}S_{L}(\rho_{AC})+\frac{1}{d_{A}}S_{L}(\rho_{BC})+\frac{d_{A}d_{B}+1-d_{A}-d_{B}}{d_{A}d_{B}}
\end{align}
and now assume that we can write the state as a product state $\rho_{ABC}=\rho_{AB}\otimes\rho_{C}$, where $\rho_{C}=|\phi\rangle\langle\phi|$ is a pure state, we can rewrite the inequality using pseudo-additivity
\begin{align}
\big(S_{L}(\rho_{AB})+S_{L}(\rho_{C})-S_{L}(\rho_{AB})S_{L}(\rho_{C})\big)+\frac{1}{d_{A}d_{B}}S_{L}(\rho_{C}) \le \\ 
\frac{1}{d_{B}}\big(S_{L}(\rho_{A})+S_{L}(\rho_{C})-S_{L}(\rho_{A})S_{L}(\rho_{C})\big)
+ \frac{1}{d_{A}}\big(S_{L}(\rho_{B})+S_{L}(\rho_{C})-S_{L}(\rho_{B})S_{L}(\rho_{C})\big) \notag \\
+ \frac{d_{A}d_{B}+1-d_{A}-d_{B}}{d_{A}d_{B}} \notag
\end{align}
and since $\rho_{C}$ is pure we have $S_{L}(\rho_{C})=0$ and it follows
\begin{align}
S_{L}(\rho_{AB}) &\le  \frac{1}{d_{B}}S_{L}(\rho_{A})+\frac{1}{d_{A}}S_{L}(\rho_{B})+\frac{(d_{A}-1)(d_{B}-1)}{d_{A}d_{B}}
\end{align}

\end{proof}

For an alternative proof see the Appendix \ref{AB1}.

It turns out that this inequality is tighter than the inequality from Theorem \ref{LEt2}.

\begin{prop}\label{ISp1}
	The inequality from Lemma \ref{ISl1}. is always sharper than the inequality in Theorem \ref{LEt2}. from \cite{Appel2017}.
\end{prop}

For a proof see the Appendix \ref{AB2}.

\section{Dimensionally sharp subadditivity}\label{DSSA}

We now introduce a new inequality for the linear entropy. It is a dimension dependent inequality, which is sharper than subadditivity and in fact turns out to be the sharpest possible for every dimensions of a bipartite system. Therefore we call it dimensionally sharp subadditivity (DSSA). 

\begin{theorem}\label{DSSAt1}
	Let $\rho_{AB}\in \mathcal{H}_A\otimes\mathcal{H}_B$ be the state of a bipartite quantum system. Then the following inequality holds
	
	\begin{align}
	S_{L}(\rho_{AB}) &\le 
	\ S_{L}(\rho_{A})+ S_{L}(\rho_{B}) \label{DSSAt11} \\
	& -2D_AD_B\big(1-\sqrt{1-\frac{S_{L}(\rho_{A})}{D_A} }\big)\big(1-\sqrt{1-\frac{S_{L}(\rho_{B})}{D_B}}\big) \nonumber
	\end{align}
	under the assumption that
		
	\begin{align}
	S_{L}(\rho_{A})\le D_A\big(\frac{S_{L}(\rho_{B})}{D_B}-1+2\sqrt{1-\frac{S_{L}(\rho_{B})}{D_B}}\big) \label{DSSAt12}
	\end{align}
	where $D_A=\frac{d_A-1}{d_A}$ and $D_B=\frac{d_B-1}{d_B}$.
\end{theorem}

To prove this theorem, we first need some intermediate results, where we would like to point out that they are interesting in their own right.\\

\begin{lemma}\label{DSSAl1}
	Let $\big\{X_i\big\}_{i=0}^{d^2-1}$ be a basis for the Hilbert space $\mathcal{H}$, with Hermitian traceless operators (except $X_0=\mathds{1}_d$) and $\Tr(X_i X_j)=d\delta_{ij}$.\\
	
	Then $\sqrt{d-1}\ \mathds{1}_d\pm X$ is a positive semi-definite operator.
\end{lemma}

\begin{proof}
	For $X_0=\mathds{1}_d$ the statement is obvious, since $d\ge2$.\\
	Let now $\mu_k$ be the eigenvalues of a given $X_i$, $i\ne0$.\\ 
	We know that $\sum\limits_{k=0}^{d-1}\mu_k=0=\Tr(X_i)$ and $\sum\limits_{k=0}^{d-1}\mu_k^2=d=\Tr(X^2_i)$.\\
	If we optimise any $\mu_k$ (with Lagrange) we get $\mu_k=\pm\sqrt{d-1}$ and therefore in general $\pm\mu_k\le\sqrt{d-1}$. So we know that $\sqrt{d-1}\ \mathds{1}_d\pm X$ is a positive semi-definite operator.\\
\end{proof}

\begin{lemma}\label{DSSAl2}
	Let $\mathcal{H}_A$ and $\mathcal{H}_B$ be Hilbert spaces with dimension $d_A$ and $d_B$ respectively. Define a basis for each Hilbert space by $\big\{X_i\big\}_{i=0}^{d_A^2-1}$ for the Hilbert space $\mathcal{H}_{A}$ and $\big\{Y_j\big\}_{j=0}^{d_B^2-1}$ the Hilbert space $\mathcal{H}_{B}$.
	We require all bases to consist of Hermitian traceless operators (except $X_0=\mathds{1}_{d_A},\ Y_0=\mathds{1}_{d_B}$) that fulfil $\Tr(X_i X_j)=d_A\delta_{ij}$ and $\Tr(Y_i Y_j)=d_B\delta_{ij}$ respectively.\\

	Then for all such basis elements $X_i$ and $Y_j$ the inequality
	\begin{equation}
	|\langle X_i \otimes Y_j \rangle| \ge \sqrt{d_B-1}|\langle X_i \otimes \mathds{1}\rangle| +\sqrt{d_A-1}|\langle \mathds{1}\otimes Y_j \rangle|-\sqrt{d_A-1}\sqrt{d_B-1}
	\end{equation}
	holds.
\end{lemma}

\begin{proof}
Pick a basis element of each basis, denoted by $X$ and $Y$.\\
By the previous lemma $\sqrt{d_A-1}\ \mathds{1}_{d_A}\pm X$ and $\sqrt{d_B-1}\ \mathds{1}_{d_B}\pm Y$ are positive semidefinite, and so is their product. Hence we can write
\begin{align}
0\le\bigg\langle(\sqrt{d_A-1}\mathds{1}\pm X)\otimes(\sqrt{d_B-1}\mathds{1}\pm Y)\bigg\rangle
\end{align}

And we can derive four different inequalities
\begin{align}
0 \le & \sqrt{d_A-1}\sqrt{d_B-1}+\sqrt{d_B-1}\langle X\otimes \mathds{1}\rangle +\sqrt{d_A-1}\langle \mathds{1}\otimes Y\rangle+\langle X\otimes Y\rangle\\
0 \le & \sqrt{d_A-1}\sqrt{d_B-1}-\sqrt{d_B-1}\langle X\otimes \mathds{1}\rangle +\sqrt{d_A-1}\langle \mathds{1}\otimes Y\rangle-\langle X\otimes Y\rangle\\
0 \le & \sqrt{d_A-1}\sqrt{d_B-1}+\sqrt{d_B-1}\langle X\otimes \mathds{1}\rangle -\sqrt{d_A-1}\langle \mathds{1}\otimes Y\rangle-\langle X\otimes Y\rangle\\
0 \le & \sqrt{d_A-1}\sqrt{d_B-1}-\sqrt{d_B-1}\langle X\otimes \mathds{1}\rangle -\sqrt{d_A-1}\langle \mathds{1}\otimes Y\rangle+\langle X\otimes Y\rangle
\end{align}

Dividing by $\sqrt{d_A-1}\sqrt{d_B-1}$ and renaming $a=\frac{\langle X\otimes \mathds{1}\rangle}{\sqrt{d_A-1}}$, $b=\frac{\langle \mathds{1}\otimes Y\rangle}{\sqrt{d_B-1}}$ and $c=\frac{\langle X\otimes Y\rangle}{\sqrt{d_A-1}\sqrt{d_B-1}}$ we can rewrite these inequalities as
\begin{align}
0 \le & 1+a+b+c \label{DSSAl21} \\
0 \le & 1-a+b-c \label{DSSAl22} \\
0 \le & 1+a-b-c \label{DSSAl23} \\
0 \le & 1-a-b+c \label{DSSAl24}
\end{align}

Assume that $a,b,c$ are nonzero (if not, one can get the same result easily).Then there are eight possible combinations of the signs of $a,b,c$. If
\[sgn(c)=sgn(a)\cdot sgn(b)\]
then we can use one of the equations (\ref{DSSAl21})-(\ref{DSSAl24}) to get
\begin{align}
0 \le 1-|a|-|b|+|c| \label{DSSAl25}
\end{align}

Now lets look at the cases where this is not fulfilled.\\
\textbf{Case 1:} $a<0$, $b>0$, $c>0$
\[0 \le  1+a-b-c= 1-|a|-|b|-|c|\le 1-|a|-|b|+|c|\]
\textbf{Case 2:} $a>0$, $b<0$, $c>0$
\[0 \le  1-a+b-c= 1-|a|-|b|-|c|\le 1-|a|-|b|+|c|\]
\textbf{Case 3:} $a>0$, $b>0$, $c<0$
\[0 \le  1-a-b+c= 1-|a|-|b|-|c|\le 1-|a|-|b|+|c|\]
\textbf{Case 4:} $a<0$, $b<0$, $c<0$
\[0 \le  1+a+b+c= 1-|a|-|b|-|c|\le 1-|a|-|b|+|c|\]
\hspace{\fill}

So inequality (\ref{DSSAl25}) always holds and therefore we get
\begin{align}
|\langle X\otimes Y\rangle| \ge \sqrt{d_B-1}|\langle X\otimes \mathds{1}\rangle| +\sqrt{d_A-1}|\langle \mathds{1}\otimes Y\rangle|-\sqrt{d_A-1}\sqrt{d_B-1}
\end{align}

\end{proof}

\begin{lemma}\label{DSSAl3}
	For every bipartite quantum state $\rho_{AB}\in\mathcal{H}_A\otimes\mathcal{H}_B$ there exists a basis $\big\{X_i\otimes Y_j\big\}_{i,j}$, such that
	\begin{itemize}
		\item $\big\{X_i\big\}_{i}$ and $\big\{Y_j\big\}_{j}$ are local bases for $\mathcal{H}_A$ and $\mathcal{H}_B$ respectively
		\item all basis  elements consist of Hermitian traceless operators (except $X_0=\mathds{1}_{d_A},\ Y_0=\mathds{1}_{d_B}$)
		\item  $\Tr(X_i X_j)=d_A\delta_{ij}$ and $\Tr(Y_i Y_j)=d_B\delta_{ij}$
		\item $\rho_{AB}$ can be expressed by
		\begin{equation}
		\rho_{AB} = \frac{1}{d_A d_B} \big(\mathds{1} + 
		b_{10} X_1 \otimes \mathds{1} + b_{01} \mathds{1} \otimes
		Y_1+ \sum_{i,j} b_{ij} X_i \otimes Y_j \big)
		\end{equation}
		where $b_{i0} = \langle X_i \otimes \mathds{1} \rangle$, $b_{0j} = \langle \mathds{1} \otimes
		Y_j \rangle$ and $b_{ij}=\langle X_i \otimes Y_j  \rangle$.\\
	\end{itemize}
	In particular, in such a basis the local Bloch vectors are expressible by single basis elements instead of linear combinations of basis elements.
\end{lemma}

\begin{proof}
	By Theorem \ref{AA1t2}., each state can be expressed in the well-known Bloch decomposition by
	\begin{equation}
	\rho = \frac{1}{d} \big(\mathds{1} + 
	\sum_{i\ge1} b_i X_i \big)=\frac{1}{d} \big(\mathds{1} + 
	\vec{\mathbf{b}}\cdot\vec{\mathbf{\Gamma}} \big)
	\end{equation}
	Changing the operator basis $\big\{X_i\big\}_{i}$ corresponds to a rotation of the Bloch vector, which can be realised by an orthogonal matrix $Q\in SO(d^2-1)$. I.e. take $X_1=\rho-\frac{\mathds{1_d}}{d}$, which is traceless and Hermitian, normalise it and choose it as the first basis element and then complete the basis in an appropriate way.\\
	Do this for both partitions, $X_1:=\rho_A-\frac{\mathds{1_{d_A}}}{d_A}$ and $Y_1:=\rho_B-\frac{\mathds{1_{d_B}}}{d_B}$, and denote the bases by $\big\{X_i\big\}_{i}$ and $\big\{Y_j\big\}_{j}$ respectively. Then $\big\{X_i\otimes Y_j\big\}_{i,j}$ is a basis for the joint system and we can write the state as
    \begin{equation}\label{DSSAL3eq2}
    \rho_{AB}=\frac{1}{d_{AB}}\bigg(\mathds{1}_{d_A}\otimes\mathds{1}_{d_B}+\sum\limits_{i=1}^{d_A^2-1}b_{i0}X_i\otimes\mathds{1}_{d_B}+\sum\limits_{j=1}^{d_B^2-1}b_{0j}\mathds{1}_{d_B}\otimes Y_j+\sum\limits_{i,j\ge1}b_{ij}X_i\otimes Y_j\bigg).
    \end{equation}
	We changed the Bloch basis locally, such that the marginal state is expressible by a single basis element, and therefore\\
	\begin{equation}
	\rho_{AB} =\frac{1}{d_A d_B} \bigg(\mathds{1} + 
	b_{10} X_1 \otimes \mathds{1} + b_{01} \mathds{1} \otimes
	Y_1+ \sum_{i,j} b_{ij} X_i \otimes Y_j \bigg)
	\end{equation}
\end{proof}

\begin{lemma}\label{DSSAl4}
	For the local Bloch vectors q-norm $\|C_{A}\|_q:=\sqrt[q]{\sum_{i}|\langle X_i \otimes \mathds{1} \rangle|^q}$, $\|C_{B}\|_q:=\sqrt[q]{\sum_{j}|\langle \mathds{1}\otimes Y_j \rangle|^q}$ and the correlation tensor norm $\|C_{AB}\|_q:=\sqrt[q]{\sum_{i,j}|\langle X_i \otimes Y_j \rangle|^q}$ of a given bipartite state the following holds for $q\ge1$
	\begin{align}
	\|C_{AB}\|_q\ge\sqrt{d_B-1}\|C_{A}\|_q+\sqrt{d_A-1}\|C_{B}\|_q-\sqrt{d_A-1}\sqrt{d_B-1}
	\end{align}
\end{lemma}

\begin{proof}
	We choose an appropriate basis for the Hilbert space $\mathcal{H}_A\otimes\mathcal{H}_B$, according to Lemma \ref{DSSAl3}. Then we can express the correlation tensors of the joint system as $\|C_{A}\|_q=|\langle X_1\otimes\mathds{1}\rangle |$ and $\|C_{B}\|_q=|\langle \mathds{1}\otimes Y_1\rangle |$.\\
	It follows $\|C_{AB}\|_q=\sqrt[q]{\langle X_1\otimes Y_1\rangle^{q}+\sum\limits_{i,j\ge2}\langle X_{i}\otimes Y_{j}\rangle^{q}}\ge|\langle X_1\otimes Y_1\rangle |$ and by using Lemma \ref{DSSAl2}. we can therefore conclude
	\begin{align}
	\|C_{AB}\|_q\ge\sqrt{d_B-1}\|C_{A}\|_q+\sqrt{d_A-1}\|C_{B}\|_q-\sqrt{d_A-1}\sqrt{d_B-1}
	\end{align}
\end{proof}

\begin{cor}\label{DSSAc1}
	The following inequality holds true	
	\begin{align*}
	\sqrt{(d_A-1)(d_B-1)-d_A d_B\ S_{L}(\rho_{AB})+d_A\ S_{L}(\rho_{A})+d_B\ S_{L}(\rho_{B})} \tag{I}\ge\\
	\sqrt{d_B-1}\sqrt{d_A-1-d_A\ S_{L}(\rho_{A})}+\sqrt{d_A-1}\sqrt{d_B-1-d_B\ S_{L}(\rho_{B})}-\sqrt{d_A-1}\sqrt{d_B-1} \tag{II}
	\end{align*}
\end{cor}

\begin{proof}
From  (\ref{AA1t22}) and Eq. (\ref{DSSAL3eq2}) we conclude for the linear entropy of a bipartite state
\begin{align}
S_L(\rho_{AB})=1-\frac{1}{d_A d_B}\big(1+\|C_A\|_2^2+\|C_B\|_2^2+\|C_{AB}\|_2^2\big) \label{AA1t24}
\end{align}
and together with Eq. (\ref{LEeq1}) therefore we have
\begin{align}
\|C_{A}\|_2=&\sqrt{d_A-1-d_A\ S_{L}(\rho_{A})}\\
\|C_{B}\|_2=&\sqrt{d_B-1-d_B\ S_{L}(\rho_{B})}\\
\|C_{AB}\|_2=&\sqrt{(d_A-1)(d_B-1)-d_A d_B\ S_{L}(\rho_{AB})+d_A\ S_{L}(\rho_{A})+d_B\ S_{L}(\rho_{B})}
\end{align}

Now just plug these in the inequality of the previous lemma for $q=2$and the result follows.
\end{proof}

\begin{remark}
    Unfortunately Eq. (\ref{AA1t24}) does not hold for values other then $q=2$, therefore our result follows only for the linear entropy.
\end{remark}

With this last result the proof of Theorem \ref{DSSAt1}. becomes straightforward:

\begin{proof}
First we need to square both sides to get rid of the roots. Therefore we have to be sure both sides are non-negative. This obviously holds for the left hand-side.\\

Rewriting the right side of the equation (II) in Corollary \ref{DSSAc1}. and dividing it by $\sqrt{d_A-1}\sqrt{d_B-1}$ we get
\begin{align}
\sqrt{1-\frac{S_{L}(\rho_{A})}{D_A}}+\sqrt{1-\frac{S_{L}(\rho_{B})}{D_B}}-1
\end{align}
where $D_A=1-\frac{1}{d_A}$ and $D_B=1-\frac{1}{d_B}$.\\
The condition that the right hand-side (II) is non-negative then becomes
\begin{align}
\sqrt{1-\frac{S_{L}(\rho_{A})}{D_A}}+\sqrt{1-\frac{S_{L}(\rho_{B})}{D_B}}\ge1
\end{align}
or equivalently
\begin{align}
S_{L}(\rho_{B})\le D_B\big(\frac{S_{L}(\rho_{A})}{D_A}-1+2\sqrt{1-\frac{S_{L}(\rho_{A})}{D_A}}\big)
\end{align}

Assume now, that this constraint holds and that both sides of the equation in Corollary \ref{DSSAc1}. are positive. Now square both sides and divide by $d_Ad_B$. This gives

\begin{align}
D_AD_B-S_{L}(\rho_{AB})+\frac{S_{L}(\rho_{A})}{d_B}+\frac{S_{L}(\rho_{B})}{d_A} \ge\\
D_B\big(D_A-S_L(\rho_{A})\big)+D_A\big(D_B-S_L(\rho_{B})\big)+D_AD_B\notag \\
+2D_AD_B\sqrt{1-\frac{S_L(\rho_{A})}{D_A}}\sqrt{1-\frac{S_L(\rho_{B})}{D_B}}
-2D_AD_B\sqrt{1-\frac{S_L(\rho_{A})}{D_A}}-2D_AD_B\sqrt{1-\frac{S_L(\rho_{B})}{D_B}}\notag
\end{align}

which can be reformulated as
\begin{align}
S_{L}(\rho_{AB}) \le 
\ S_{L}(\rho_{A})+ S_{L}(\rho_{B}) 
-2D_AD_B\big(1-\sqrt{1-\frac{S_{L}(\rho_{A})}{D_A} }\big)\big(1-\sqrt{1-\frac{S_{L}(\rho_{B})}{D_B}}\big) \label{DSSAt13}
\end{align}

\end{proof}

\begin{remark}
	The domain where inequality (\ref{DSSAt13}) holds is constrained. This is not a consequence of this proof and indeed it is wrong outside of its assigned domain.\\
	Assume $d_A=d_B=d$ and $\rho_{AB}=\frac{\mathds{1}_{d^2}}{d^2}$ is the completely mixed state and therefore $S_{L}(\rho_{AB})=1-\frac{1}{d^2}$. But we also have $\rho_{A}=\rho_{B}=\frac{\mathds{1}_{d}}{d}$  and $S_{L}(\rho_{A})=S_{L}(\rho_{B})=1-\frac{1}{d}$.\\ Inequality (\ref{DSSAt13}) becomes $1-\frac{1}{d^2}\le0$, which can not hold since $d\ge2$.
\end{remark}
\begin{remark}
	By taking the limits $d_A=d_B \rightarrow \infty$ DSSA converges towards the inequality from Lemma \ref{LEl1}.
	\begin{align}
	S_{L}(\rho_{AB})\le& \ S_{L}(\rho_{A})+ S_{L}(\rho_{B})\\
	&-2(1-\sqrt{1- S_{L}(\rho_{A})})(1-\sqrt{1- S_{L}(\rho_{B})})\notag\\
	\notag
	\end{align}
\end{remark}

\section{Sharpest bound}\label{SB}
We want to combine inhomogeneous subadditivity with dimensionally sharp subadditivity and prove this "combination" to give a tight upper bound for the linear entropy of a bipartite quantum system.\\
From now on, $S_L(\rho_{A})$, $S_L(\rho_{B})$ and $S_L(\rho_{AB})$ will be denoted by the coordinates $x,y,z$. We will also use the notation $D_A:=\frac{d_A-1}{d_A}$ and $D_B:=\frac{d_B-1}{d_B}$.\\

\begin{definition}
Let's define the functions of the inequalities. The inequality in Lemma \ref{ISl1}. gives
	\begin{align}
	h(x,y):=\frac{1}{d_{B}}x+\frac{1}{d_{A}}y+D_AD_B
	\end{align}
	and inequality (\ref{DSSAt11}) in Theorem \ref{DSSAt1}.
	\begin{align}
	g(x,y):=x+y-2D_AD_B(1-\sqrt{1-\frac{x}{D_A}})(1-\sqrt{1-\frac{y}{D_B}}\big)
	\end{align}
	The restriction (\ref{DSSAt12}) then becomes
	\begin{align}
	r(y):=D_A\big(\frac{1}{D_B}y-1+2\sqrt{1-\frac{y}{D_B}}\big)
	\end{align}
\end{definition}

So we can rewrite the inhomogeneous subadditivity as 
\begin{equation}
    S_L(\rho_{AB})\le h(S_L(\rho_{A}),S_L(\rho_{B}))
\end{equation}
and the dimensionally sharp subadditivity as
\begin{equation}
    S_L(\rho_{AB})\le g(S_L(\rho_{A}),S_L(\rho_{B}))
\end{equation}
This two inequalities can be combined by the following proposition.

\begin{prop}\label{SBp1}
	The function
	\begin{align}
	f(x,y):=\begin{cases}
	g(x,y)\qquad x\le r(y)\\
	h(x,y)\qquad x> r(y)
	\end{cases}
	\end{align}
	is continuously differentiable.
\end{prop}

For a proof see Appendix \ref{AC}.

\begin{figure}[htbp]
	\includegraphics[trim = 0mm 0mm 0mm 0mm, clip, width=1\textwidth]{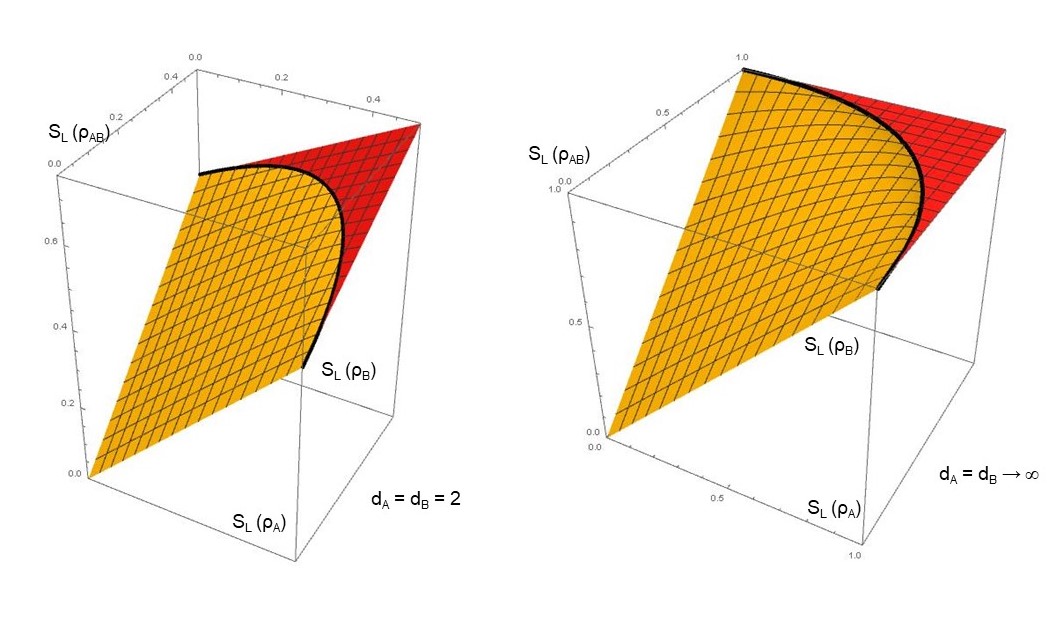} 
	\caption{\small \sl Recall the definition of the entropy vector and that for a bipartite system this vector has dimension $2^2-1=3$. For the linear entropy the set of all entropy vectors (=entropy body) is a subset of the box $\{(x,y,z)|\ 0\le x\le D_A,\ 0\le y\le D_B,\ 0\le z\le D_{AB}\}$. The image shows our two new inequalities, the inhomogeneous subadditivity and the dimensionally sharp subadditivity, which give an additional constraint. The entropy body lies below the pictured surface. The equality in DSSA is plotted in orange and equality in ISA in red. Both surfaces are transparently continued beyond where they hold and are sharp. The figure on the left is for the case $d_A=d_B=2$ and the figure on the right for the case $d_A=d_B\rightarrow\infty$.}
\end{figure}

\begin{lemma}\label{SBl1}
    For arbitrary $0\le a\le D_A$ and $0\le b\le D_B$ there exists a state such that
    \begin{align}
        S_L(\rho_A)=a\\
        S_L(\rho_B)=b\\
        S_L(\rho_{AB})=f(a,b)\\
    \end{align}
\end{lemma}

For the proof see Appendix \ref{AD}.

This means that every point of the restricting surface is an entropy vector, which implies the following theorem.

\begin{theorem}\label{SBt1}
	The inequality defined by
	\begin{align}
	S_L(\rho_{AB})\le f(S_L(\rho_{A}),S_L(\rho_{B}))
	\end{align}
	holds for a bipartite quantum system and is tight.\\
\end{theorem}

\section{R\'{e}nyi 2-entropy and Purity}\label{Re}

The R\'enyi 2-entropy, sometimes also called the collision entropy, is strongly related to the linear entropy by Equation (\ref{Tsallistorenyi}). With this, the above inequalities can be translated into inequalities for the R\'enyi 2-entropy. This is not in contrast to the results of \cite{Linden2013}, since the inequality will not be homogeneous.
\par
The linear entropy and the R\'{e}nyi 2-entropy depend on each other by
\begin{align}
S^2(\rho)&:=-\log_2(\Tr(\rho^2))=-\log_2(1-S_L(\rho)) \label{461}\\
S_L(\rho)&:=1-\Tr(\rho^2)=1-2^{-S^2(\rho)}. \label{462}
\end{align}

Therefore Theorem \ref{SBt1}. can be reformulated to an inequality for the R\'enyi 2-entropy.

\begin{prop}\label{RPp1}
	Define the following function
	\begin{align}
	f_R(x,y):=\begin{cases}
	g_R(x,y)\qquad x\le r_R(y)\\
	h_R(x,y)\qquad x> r_R(y)
	\end{cases}\label{RPp11}
	\end{align}
	where
	
    \begin{align}
    h_R(x,y):= & -\log_2\big(\frac{2^{-x}}{d_B} +\frac{2^{-y}}{d_A}-\frac{1}{d_A d_B}\big)\\
    g_R(x,y):= & -\log_2\bigg(2^{-x}+2^{-y}-1+ 2\frac{d_A-1}{d_A}\frac{d_B-1}{d_B}\big(1-\sqrt{\frac{d_A\ 2^{-x}-1}{d_A-1}}\big)\big(1-\sqrt{\frac{d_B\ 2^{-y}-1}{d_B-1}}\big)\bigg)\\
	r_R(y):= & -\log_2\bigg(\frac{d_A-1}{d_A}\bigg(\frac{d_B}{d_B-1}\  2^{-y}-\frac{1}{d_B-1}-2\sqrt{\frac{d_B\ 2^{-y}-1}{d_B-1}}\bigg)-1\bigg)
	\end{align}
	
	Then $f_R(x,y)$ is a continuously differentiable function and the following inequality 
	\begin{align}
	S^2(\rho_{AB})\le f_R(S^2(\rho_{A}),S^2(\rho_{B}))
	\end{align}
	holds for a bipartite quantum system.
\end{prop}

For the proof see Appendix \ref{AE}.

These inequalities may look somewhat unhandy, but it can be more useful for certain applications, since the R\'enyi entropy is  additive for product states.

\begin{figure}[htbp]
	\includegraphics[trim = 0mm 40mm 0mm 0mm, clip, width=0.8\textwidth]{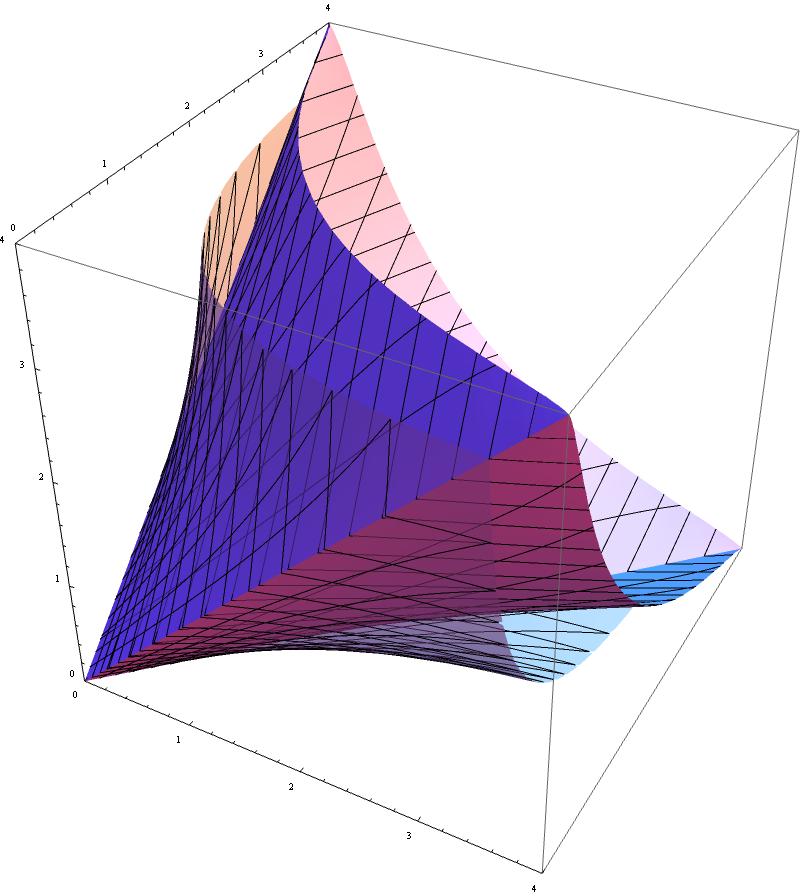} 
	\caption{\small \sl  Assume we have a tripartite system in a \textbf{pure} state. The entropy vector then has $2^3-1=7$ entries, which is obviously hard to plot! But we know that $S(ABC)=0$, $S(A)=S(BC)$, $S(B)=S(AC)$ and $S(C)=S(AB)$, which leaves us with only $3$ different entries.\newline
	Let now the dimensions be $d_A=d_B=d_C=16$.\newline
	The figure shows the entropy cone (inner opaque surface) resulting from the 3 different inequalities, $S(A)\le f_R\big(S(B),S(C)\big)$, $S(B)\le f_R\big(S(A),S(C)\big)$ and $S(C)\le f_R\big(S(A),S(B)\big)$, where $f_R(x,y)$ is the same as in Eq. (\ref{RPp11}).\newline
	The outer transparent cone is obtained in an analogous way from the subadditivity of the linear entropy.}
\end{figure}

The purity of a quantum state $\rho$ is defined as $\gamma(\rho)=\Tr(\rho^2)$ and therefore we can express the inequalities for the linear entropy also in terms of the purity.

\begin{cor}
    Define the function
    \begin{align}
	f_P(x,y):=\begin{cases}
	\frac{x}{d_B}+\frac{y}{d_A}-\frac{1}{d_Ad_B}\qquad & x< 1-\frac{d_A-1}{d_A}\bigg(\frac{1-d_B y}{d_B-1}+2\sqrt{\frac{d_B y-1}{d_B-1}}\bigg)\\
	x+y-1+2\frac{d_A-1}{d_A}\frac{d_B-1}{d_B}\bigg(1-\sqrt{\frac{d_Ax}{d_A-1}}\bigg)\bigg(1-\sqrt{\frac{d_B y}{d_B-1}}\bigg)\qquad & x\ge 1-\frac{d_A-1}{d_A}\bigg(\frac{1-d_B\gamma(B)}{d_B-1}+2\sqrt{\frac{d_B\gamma(B)-1}{d_B-1}}\bigg)
	\end{cases}
	\end{align}
	
	Then $f_P(x,y)$ is a continuously differentiable function and the following inequality 
	\begin{align}
	\gamma(\rho_{AB})\ge f_P(\gamma(\rho_{A}),\gamma(\rho_{B}))
	\end{align}
	holds for the purities of a bipartite quantum system.
\end{cor}

\begin{proof}
This follows immediately from Theorem \ref{SBt1}.
\end{proof}

\section{Inverted inequality}\label{II}

It is well known how one can use purification to obtain the lower bound to a bipartite entropy (Araki-Lieb inequality) from its upper bound given by subadditivity. We therefore call this lower bound the inverted inequality of subadditivity.\\
By using the same argument, a lower bound for the linear entropy following from Theorem \ref{SBt1}. can be constructed. By purification, for every state $\rho$ there exists a pure state $|\phi\rangle\langle\phi|$ with $\rho$ as the reduced density operator. Now assume that $\rho=\rho_{AB}$ is a bipartite state and therefore $|\phi\rangle\langle\phi|$ is a tripartite state. Since it is pure, any two partitions have the same entropy and we can use this property to transform every upper bound inequality in a lower bound inequality. \\
Proposition \ref{SBp1}. states, that the function $f(S_L(\rho_{A}),S_L(\rho_{B})$ is continuously differentiable. Therefore, using the implicit function theorem and the inverse function theorem (see Appendix \ref{AF}), the existence of an inverted inequality can be shown.\\
%Also the explicit formulas are derived, but they are a bit longish.\\

\begin{lemma}\label{IIl1}
	For $d_A,\ d_B < \infty$ the function $f(x,y)$ defined in Proposition \ref{SBp1},
	\begin{align}
	h(x,y)=&\frac{1}{d_{B}}x+\frac{1}{d_{A}}y+D_AD_B\\
	g(x,y)=&x+y-2D_AD_B(1-\sqrt{1-\frac{x}{D_A}})(1-\sqrt{1-\frac{y}{D_B}}\big)\\
	r(y)=&D_A\big(\frac{1}{D_B}y-1+2\sqrt{1-\frac{y}{D_B}}\big)\\
	f(x,y)=&\begin{cases}
	g(x,y)\qquad x\le r(y)\\
	h(x,y)\qquad x> r(y)
	\end{cases}
	\end{align}
	is strictly monotonically increasing in both variables.
\end{lemma}

\begin{proof}
	For $h(x,y)$ this holds, since it is linear in both variables.\\
	The gradient of $g(x,y)$, (\ref{AC1}) in the proof of Proposition \ref{SBp1}., is $\nabla g(x,y)|_{(0,0)}=\begin{pmatrix} 1\\ 1 \end{pmatrix}$ at the origin and decreases for growing $x,y$. But, since on the boundary $\Omega$ the Gradient $\nabla g(x,y)|_{\Omega}=\begin{pmatrix} \frac{1}{d_B}\\ \frac{1}{d_A} \end{pmatrix}$ (\ref{AC2}) is still positive, $g(x,y)$ is a strictly monotonically increasing function for $x\le r(y)$.
\end{proof}

\begin{prop}\label{IIp1}
	The inequality from Theorem \ref{SBt1}. can be inverted to two new inequalities.
\end{prop}

\begin{proof}
	Assume we have a system in state $\rho_{AB} \in \mathcal{H}_A\otimes\mathcal{H}_B$. If we take a copy of this Hilbert space (denoted by $\mathcal{H}_R$), there exists a pure state  $|\phi\rangle\langle\phi|\in(\mathcal{H}_A\otimes\mathcal{H}_B)^2=\mathcal{H}_A\otimes\mathcal{H}_B\otimes\mathcal{H}_R$ such that $\Tr_R(|\phi\rangle\langle\phi|)=\rho_{AB}$. Denote by $\rho_R=\Tr_{AB}(|\phi\rangle\langle\phi|)$.\\
	
	For a pure bipartite state the entropies of the reduced states from the two subsystems are equal, $S_L(\rho_{AR})=S_L(\rho_{B})$, $S_L(\rho_{BR})=S_L(\rho_{A})$  and $S_L(\rho_{R})=S_L(\rho_{AB})$.\\
	Every inequality for the linear entropy can therefore be transformed into an inverted inequality.\\
	
	Let $f(x,y)$ denote the continuously differentiable function from Proposition \ref{SBp1}.
	For the system $AR$ in the state $\rho_{AR}\in\mathcal{H}_A\otimes\mathcal{H}_R$, Theorem \ref{SBt1}. states
	\begin{equation}
	S_L(\rho_{AR})\le f(S_L(\rho_{A}),S_L(\rho_{R}))
	\end{equation}
	and, by the previous discussion
	\begin{equation}
	S_L(\rho_{B})\le f(S_L(\rho_{A}),S_L(\rho_{AB})) \label{IIp11}
	\end{equation}
	For the other inequality just take the system $BR$ in the state $\rho_{BR}\in\mathcal{H}_B\otimes\mathcal{H}_R$ and, by the same argument
	\begin{equation}
	S_L(\rho_{A})\le f(S_L(\rho_{B}),S_L(\rho_{AB})) \label{IIp12}
	\end{equation}
\end{proof}

\begin{prop}\label{IIp2}
	There exists continuously differentiable functions $G_1$ and $G_2$, such that the inequalities (\ref{IIp11}) and (\ref{IIp12}) from the proof of the previous proposition can be written in the form
	\begin{align}
	S_L(\rho_{AB})\ge&\ G_1(S_L(\rho_{A}),S_L(\rho_{B})) \label{IIp21}\\
	S_L(\rho_{AB})\ge&\ G_2(S_L(\rho_{A}),S_L(\rho_{B})) \label{IIp22}
	\end{align}
\end{prop}

\begin{proof}
	For equality we can rewrite (\ref{IIp11}) in the form
	\begin{equation}
	F(x,y,z):=f(x,z)-y=0
	\end{equation} 
	Since this is linear in $y$ and by Lemma \ref{IIl1}, we know that all the partial derivatives are always positive. Therefore, the implicit function theorem can be used and $U$ can be chosen to be the whole domain of $f(x,y)$. Setting $\mathbf{v}=(x,y)$ and $\mathbf{w}=z$ the theorem states, that for every $(x_0,y_0)$ there exists a continuously differentiable function $G_1(x,y)$, such that $G_1(x_0,y_0)=z$ and $F\big(x,y,G(x,y)\big)=0$.
	Therefore we know that
	\begin{equation}
	f(x,G_1(x,y))=y \label{IIp2p1}
	\end{equation}
	Now fix $x$ to $x_0$. Define
	\begin{equation}
	f_{x_0}(y):= f(x,y)|_{x=x_0}
	\end{equation}
	and note that $\frac{\partial f(x,y)}{\partial y}\ge \frac{1}{d_Ad_R}>0$. By the inverse function theorem we can invert $f_{x_0}(y)$ and therefore rewrite Eq. (\ref{IIp2p1}) as
	\begin{equation}
	G_1(x_0,y))=f_{x_0}^{-1}(y)
	\end{equation}
	Since this holds for every $x_0$, $G_1(x,y)$ is the continuously differentiable function of Eq. (\ref{IIp21}).\\
	By an analogous argument we show the same for $G_2(x,y)$ of Eq. (\ref{IIp22}).\\
\end{proof}

Note, that this is not the same as to invert $f(x,y):\mathds{R}^2\rightarrow\mathds{R}$, which is not possible.\\

So we have proven the existence and some properties of the functions $G_1$ and $G_2$, which give lower bound inequalities from purification. Let us now derive them explicitly. For the inhomogeneous subadditivity they are easy to determine, the two inverted inequalities following from the inequality from Lemma \ref{ISl1}. are
\begin{align}
S_{L}(\rho_{AB})\ge &\ d_A S_{L}(\rho_{B})-\frac{1}{d_B}S_L(\rho_{A})- \frac{(d_Ad_B-1)(d_A-1)}{d_Ad_B}\\
=: &\tilde{h}_1\big(S_{L}(\rho_{A}),S_{L}(\rho_{B})\big)\notag \\
\notag \\
S_{L}(\rho_{AB})\ge &\ d_B S_{L}(\rho_{A})-\frac{1}{d_A}S_L(\rho_{B})- \frac{(d_Ad_B-1)(d_B-1)}{d_Ad_B}\\
=: &\tilde{h}_2\big(S_{L}(\rho_{A}),S_{L}(\rho_{B})\big)\notag
\end{align}

In the case of the dimensionally sharp subadditivity, things get a bit more complicated and one always has to be careful with the restriction. The inverted inequalities following from inequality (\ref{DSSAt11}) are given by
{\allowdisplaybreaks
\begin{align}
	S_{L}(\rho_{AB})\ge & D_{AB}\bigg(1-\bigg(D_A\bigg(1-\sqrt{1-\frac{S_L(\rho_{A})}{D_A}}\bigg)\\
	+&\sqrt{\frac{S_L(\rho_{A})-S_L(\rho_{B})}{D_{AB}} +\bigg((1-D_A)+D_A\sqrt{1-\frac{S_L(\rho_{A})}{D_A}}\bigg)^2}\bigg)^2\bigg)\notag \\
	=& S_L(\rho_{B})-S_L(\rho_{A})+2D_AD_{AB}\bigg(1-\sqrt{1-\frac{S_L(\rho_{A})}{D_A}}\bigg)\notag \\
	&\bigg(1-2D_A-\sqrt{\frac{S_L(\rho_{A})-S_L(\rho_{B})}{D_{AB}}+\bigg((1-D_A)+D_A\sqrt{1-\frac{S_L(\rho_{A})}{D_A}}\bigg)^2}\bigg)\notag \\
	=: &\tilde{g}_1\big(S_{L}(\rho_{A}),S_{L}(\rho_{B})\big)\notag
\end{align}}
\hspace{\fill}
	
and
{\allowdisplaybreaks	
\begin{align}
S_{L}(\rho_{AB})\ge & D_{AB}\bigg(1-\bigg(D_B\bigg(1-\sqrt{1-\frac{S_L(\rho_{B})}{D_B}}\bigg)\\
+&\sqrt{\frac{S_L(\rho_{B})-S_L(\rho_{A})}{D_{AB}} +\bigg((1-D_B)+D_B\sqrt{1-\frac{S_L(\rho_{A})}{D_B}}\bigg)^2}\bigg)^2\bigg)\notag \\
=& S_L(\rho_{A})-S_L(\rho_{B})+2D_BD_{AB}\bigg(1-\sqrt{1-D_BS_L(\rho_{B})}\bigg)\notag \\
&\bigg(1-2D_B-\sqrt{\frac{S_L(\rho_{B})-S_L(\rho_{A})}{D_{AB}}+\bigg((1-D_B)+D_B\sqrt{1-\frac{S_L(\rho_{B})}{D_B}}\bigg)^2}\bigg)\notag \\
=: &\tilde{g}_2\big(S_{L}(\rho_{A}),S_{L}(\rho_{B})\big)\notag
\end{align}}
and hold for
\begin{align}
S_{L}(\rho_{AB})\le D_{AB}\big(\frac{S_{L}(\rho_{A})}{D_A}-1+2\sqrt{1-\frac{S_{L}(\rho_{A})}{D_A}}\big)=: &\tilde{r}_1\big(S_{L}(\rho_{A})\big)
\end{align}
and
\begin{align}
S_{L}(\rho_{AB})\le D_{AB}\big(\frac{S_{L}(\rho_{B})}{D_B}-1+2\sqrt{1-\frac{S_{L}(\rho_{B})}{D_B}}\big)=: &\tilde{r}_2\big(S_{L}(\rho_{A})\big)
\end{align}
respectively.\\

Now define
\begin{align}
\tilde{f}_1(x,y):=&\begin{cases}
\tilde{g}_1(x,y)\qquad z\le \tilde{r}_1(x)\\
\tilde{h}_1(x,y)\qquad z> \tilde{r}_1(x)
\end{cases}
\end{align}
and
\begin{align}
\tilde{f}_2(x,y):=&\begin{cases}
\tilde{g}_2(x,y)\qquad z\le \tilde{r}_2(x)\\
\tilde{h}_2(x,y)\qquad z> \tilde{r}_2(x)
\end{cases}
\end{align}
By Proposition \ref{SBp1}. and the implicit function theorem, we know that $\tilde{f}_1(x,y)$ and $\tilde{f}_2(x,y)$ are continuously differentiable functions.\\
Now define
\begin{align}
\tilde{f}(x,y):=&\begin{cases}
\tilde{f}_1(x,y)\qquad x\le y\\
\tilde{f}_2(x,y)\qquad x> y
\end{cases}
\end{align}

The above discussion leads to

\begin{cor}
The inequality
\begin{equation}
S_{L}(\rho_{AB})\ge \tilde{f}\big(S_{L}(\rho_{A}),S_{L}(\rho_{B})\big)
\end{equation}
holds and is sharper than the Araki-Lieb inequality for the Tsallis 2-entropy.
\end{cor}
	
\begin{proof}
The first claim follows from the previous discussion.\\
Since the inequality from Theorem \ref{SBt1}. is sharper than subadditivity, the inverted inequality is sharper than Araki-Lieb.
\end{proof}

That inhomogeneous subadditivity and dimensionally sharp subadditivity are the sharpest possible inequalities, does however not imply, that the inverted inequality is the sharpest. Why is this so? Remember that in the proof we used our inequality for states of the form $\rho_{AR}$ (or $\rho_{BR}$). But this means that the inverted inequality is only the sharpest for states $\rho_{AR}\in\mathcal{H}_A\otimes\mathcal{H}_R$ (or $\rho_{BR}\in\mathcal{H}_B\otimes\mathcal{H}_R$) which can be purified in the space $\mathcal{H}_A\otimes\mathcal{H}_B\otimes\mathcal{H}_R$.\\

\section{Conclusion}\label{CO}
We have introduced a dimensionally sharp entropy inequality for the Tsallis 2-entropy (ISA+DSSA). It provides a sharp improvement for all finite dimensional quantum systems, extending entropy inequalities beyond the asymptotic domain. From a physical perspective, these inequalities provides stronger constraints for the distribution of purities in multipartite quantum states. As any state that is prepared experimentally, and used for information encoding, can be represented by a finite dimensional Hilbert space, the inequalities should be useful whenever one is interested in the distribution of purities or equivalently Bloch vector lengths. The lengths of Bloch vectors over different 'sectors' \cite{Jens1,Jens2} (i.e. tensor products of $k$ local Bloch operators would correspond to a 'k'-sector) or correlation tensors \cite{Julio,Zukowski, Claude1,Claude2, Weinfurter} are often used to quantify correlations and detect entanglement, so sharper inequalities will also be useful in this context. For the future we hope to find a sharp lower bound and extend the approach to multipartite systems, also finding an inequality akin to strong subadditivity.

\emph{Acknowledgements} The authors would like to acknowledge fruitful discussions with Paul Appel, Otfried G\"uhne, Felix Huber and Milan Mosonyi. MH and SM acknowledge funding from the FWF (Y879-N27, I3053-N27 and P31339-N27). CK gratefully acknowledges support from the European Research Council (“reFUEL” ERC-2017-STG 758149). CE and JS acknowledge funding from the German ResearchFoundation Project EL710/2-1 and JS from the Basque Government grant IT986-16 and grant PGC2018-101355-B-I00 (MCIU/AEI/FEDER,UE).
\bibliography{mybib}

%merlin.mbs apsrev4-1.bst 2010-07-25 4.21a (PWD, AO, DPC) hacked
%Control: key (0)
%Control: author (8) initials jnrlst
%Control: editor formatted (1) identically to author
%Control: production of article title (-1) disabled
%Control: page (0) single
%Control: year (1) truncated
%Control: production of eprint (0) enabled
\begin{thebibliography}{22}%
\makeatletter
\providecommand \@ifxundefined [1]{%
 \@ifx{#1\undefined}
}%
\providecommand \@ifnum [1]{%
 \ifnum #1\expandafter \@firstoftwo
 \else \expandafter \@secondoftwo
 \fi
}%
\providecommand \@ifx [1]{%
 \ifx #1\expandafter \@firstoftwo
 \else \expandafter \@secondoftwo
 \fi
}%
\providecommand \natexlab [1]{#1}%
\providecommand \enquote  [1]{``#1''}%
\providecommand \bibnamefont  [1]{#1}%
\providecommand \bibfnamefont [1]{#1}%
\providecommand \citenamefont [1]{#1}%
\providecommand \href@noop [0]{\@secondoftwo}%
\providecommand \href [0]{\begingroup \@sanitize@url \@href}%
\providecommand \@href[1]{\@@startlink{#1}\@@href}%
\providecommand \@@href[1]{\endgroup#1\@@endlink}%
\providecommand \@sanitize@url [0]{\catcode `\\12\catcode `\$12\catcode
  `\&12\catcode `\#12\catcode `\^12\catcode `\_12\catcode `\%12\relax}%
\providecommand \@@startlink[1]{}%
\providecommand \@@endlink[0]{}%
\providecommand \url  [0]{\begingroup\@sanitize@url \@url }%
\providecommand \@url [1]{\endgroup\@href {#1}{\urlprefix }}%
\providecommand \urlprefix  [0]{URL }%
\providecommand \Eprint [0]{\href }%
\providecommand \doibase [0]{http://dx.doi.org/}%
\providecommand \selectlanguage [0]{\@gobble}%
\providecommand \bibinfo  [0]{\@secondoftwo}%
\providecommand \bibfield  [0]{\@secondoftwo}%
\providecommand \translation [1]{[#1]}%
\providecommand \BibitemOpen [0]{}%
\providecommand \bibitemStop [0]{}%
\providecommand \bibitemNoStop [0]{.\EOS\space}%
\providecommand \EOS [0]{\spacefactor3000\relax}%
\providecommand \BibitemShut  [1]{\csname bibitem#1\endcsname}%
\let\auto@bib@innerbib\@empty
%</preamble>
\bibitem [{\citenamefont {Yeung}(2008)}]{Yeung2008}%
  \BibitemOpen
  \bibfield  {author} {\bibinfo {author} {\bibfnamefont {R.~W.}\ \bibnamefont
  {Yeung}},\ }\href {\doibase 10.1109/TIT.2009.2021372} {\emph {\bibinfo
  {title} {Springer US}}}\ (\bibinfo  {publisher} {Springer US},\ \bibinfo
  {year} {2008})\BibitemShut {NoStop}%
\bibitem [{\citenamefont {Nielsen}\ and\ \citenamefont
  {Chuang}(2010)}]{Nielsen2010}%
  \BibitemOpen
  \bibfield  {author} {\bibinfo {author} {\bibfnamefont {M.~A.}\ \bibnamefont
  {Nielsen}}\ and\ \bibinfo {author} {\bibfnamefont {I.~L.}\ \bibnamefont
  {Chuang}},\ }\href {\doibase 10.1017/CBO9780511976667} {\emph {\bibinfo
  {title} {Cambridge University Press}}}\ (\bibinfo  {publisher} {Cambridge
  University Press},\ \bibinfo {year} {2010})\ \Eprint
  {http://arxiv.org/abs/1011.1669v3} {arXiv:1011.1669v3} \BibitemShut {NoStop}%
\bibitem [{\citenamefont {Cadney}\ \emph {et~al.}(2012)\citenamefont {Cadney},
  \citenamefont {Linden},\ and\ \citenamefont {Winter}}]{Cadney12}%
  \BibitemOpen
  \bibfield  {author} {\bibinfo {author} {\bibfnamefont {J.}~\bibnamefont
  {Cadney}}, \bibinfo {author} {\bibfnamefont {N.}~\bibnamefont {Linden}}, \
  and\ \bibinfo {author} {\bibfnamefont {A.}~\bibnamefont {Winter}},\ }\href
  {\doibase 10.1109/TIT.2012.2185036} {\bibfield  {journal} {\bibinfo
  {journal} {IEEE Transactions on Information Theory}\ }\textbf {\bibinfo
  {volume} {58}},\ \bibinfo {pages} {3657} (\bibinfo {year}
  {2012})}\BibitemShut {NoStop}%
\bibitem [{\citenamefont {Pippenger}(2003)}]{Pippenger2003}%
  \BibitemOpen
  \bibfield  {author} {\bibinfo {author} {\bibfnamefont {N.}~\bibnamefont
  {Pippenger}},\ }\href@noop {} {\bibfield  {journal} {\bibinfo  {journal}
  {IEEE Transactions on Information Theory}\ }\textbf {\bibinfo {volume}
  {49}},\ \bibinfo {pages} {773} (\bibinfo {year} {2003})}\BibitemShut
  {NoStop}%
\bibitem [{\citenamefont {Linden}\ and\ \citenamefont
  {Winter}(2005)}]{Linden2005}%
  \BibitemOpen
  \bibfield  {author} {\bibinfo {author} {\bibfnamefont {N.}~\bibnamefont
  {Linden}}\ and\ \bibinfo {author} {\bibfnamefont {A.}~\bibnamefont
  {Winter}},\ }\href {\doibase 10.1007/s00220-005-1361-2} {\bibfield  {journal}
  {\bibinfo  {journal} {Communications in Mathematical Physics}\ }\textbf
  {\bibinfo {volume} {259}},\ \bibinfo {pages} {129} (\bibinfo {year}
  {2005})},\ \Eprint {http://arxiv.org/abs/0406162} {arXiv:0406162 [quant-ph]}
  \BibitemShut {NoStop}%
\bibitem [{\citenamefont {Klir}(2006)}]{Klir2006}%
  \BibitemOpen
  \bibfield  {author} {\bibinfo {author} {\bibfnamefont {G.~J.}\ \bibnamefont
  {Klir}},\ }\href@noop {} {\emph {\bibinfo {title} {{Uncertainty and
  Information: Foundations of Generalized Information Theory}}}}\ (\bibinfo
  {publisher} {Wiley-Interscience},\ \bibinfo {year} {2006})\BibitemShut
  {NoStop}%
\bibitem [{\citenamefont {Linden}\ \emph {et~al.}(2013)\citenamefont {Linden},
  \citenamefont {Mosonyi},\ and\ \citenamefont {Winter}}]{Linden2013}%
  \BibitemOpen
  \bibfield  {author} {\bibinfo {author} {\bibfnamefont {N.}~\bibnamefont
  {Linden}}, \bibinfo {author} {\bibfnamefont {M.}~\bibnamefont {Mosonyi}}, \
  and\ \bibinfo {author} {\bibfnamefont {A.}~\bibnamefont {Winter}},\ }\href
  {\doibase 10.1098/rspa.2012.0737} {\bibfield  {journal} {\bibinfo  {journal}
  {Proceedings of the Royal Society}\ }\textbf {\bibinfo {volume} {469}}
  (\bibinfo {year} {2013}),\ 10.1098/rspa.2012.0737},\ \Eprint
  {http://arxiv.org/abs/1212.0248} {arXiv:1212.0248} \BibitemShut {NoStop}%
\bibitem [{\citenamefont {Cadney}\ \emph {et~al.}(2014)\citenamefont {Cadney},
  \citenamefont {Huber}, \citenamefont {Linden},\ and\ \citenamefont
  {Winter}}]{Cadney2014}%
  \BibitemOpen
  \bibfield  {author} {\bibinfo {author} {\bibfnamefont {J.}~\bibnamefont
  {Cadney}}, \bibinfo {author} {\bibfnamefont {M.}~\bibnamefont {Huber}},
  \bibinfo {author} {\bibfnamefont {N.}~\bibnamefont {Linden}}, \ and\ \bibinfo
  {author} {\bibfnamefont {A.}~\bibnamefont {Winter}},\ }\href {\doibase
  10.1016/j.laa.2014.03.035} {\bibfield  {journal} {\bibinfo  {journal} {Linear
  Algebra and Its Applications}\ }\textbf {\bibinfo {volume} {452}},\ \bibinfo
  {pages} {153} (\bibinfo {year} {2014})},\ \Eprint
  {http://arxiv.org/abs/1308.0539} {arXiv:1308.0539} \BibitemShut {NoStop}%
\bibitem [{\citenamefont {Raggio}(1995)}]{Raggio1995}%
  \BibitemOpen
  \bibfield  {author} {\bibinfo {author} {\bibfnamefont {G.~A.}\ \bibnamefont
  {Raggio}},\ }\href@noop {} {\bibfield  {journal} {\bibinfo  {journal}
  {Journal of Mathematical Physics}\ }\textbf {\bibinfo {volume} {36}}
  (\bibinfo {year} {1995})}\BibitemShut {NoStop}%
\bibitem [{\citenamefont {Audenaert}(2007)}]{Audenaert2007}%
  \BibitemOpen
  \bibfield  {author} {\bibinfo {author} {\bibfnamefont {K.~M.~R.}\
  \bibnamefont {Audenaert}},\ }\href {\doibase 10.1063/1.2771542} {\bibfield
  {journal} {\bibinfo  {journal} {Journal of Mathematical Physics}\ }\textbf
  {\bibinfo {volume} {48}},\ \bibinfo {pages} {083507} (\bibinfo {year}
  {2007})},\ \Eprint {http://arxiv.org/abs/0705.1276v1} {arXiv:0705.1276v1}
  \BibitemShut {NoStop}%
\bibitem [{\citenamefont {Fano}(1957)}]{Fano}%
  \BibitemOpen
  \bibfield  {author} {\bibinfo {author} {\bibfnamefont {U.}~\bibnamefont
  {Fano}},\ }\href {https://doi.org/10.1103/RevModPhys.29.74} {\bibfield
  {journal} {\bibinfo  {journal} {Reviews of Modern Physics}\ }\textbf
  {\bibinfo {volume} {29}},\ \bibinfo {pages} {74} (\bibinfo {year}
  {1957})}\BibitemShut {NoStop}%
\bibitem [{\citenamefont {Zhang}\ \emph {et~al.}(2008)\citenamefont {Zhang},
  \citenamefont {Gong}, \citenamefont {Zhang},\ and\ \citenamefont
  {Guo}}]{Zhang2008}%
  \BibitemOpen
  \bibfield  {author} {\bibinfo {author} {\bibfnamefont {C.~J.}\ \bibnamefont
  {Zhang}}, \bibinfo {author} {\bibfnamefont {Y.~X.}\ \bibnamefont {Gong}},
  \bibinfo {author} {\bibfnamefont {Y.~S.}\ \bibnamefont {Zhang}}, \ and\
  \bibinfo {author} {\bibfnamefont {G.~C.}\ \bibnamefont {Guo}},\ }\href
  {\doibase 10.1103/PhysRevA.78.042308} {\bibfield  {journal} {\bibinfo
  {journal} {Physical Review A - Atomic, Molecular, and Optical Physics}\
  }\textbf {\bibinfo {volume} {78}},\ \bibinfo {pages} {1} (\bibinfo {year}
  {2008})},\ \Eprint {http://arxiv.org/abs/0806.2598v3} {arXiv:0806.2598v3}
  \BibitemShut {NoStop}%
\bibitem [{\citenamefont {Petz}\ and\ \citenamefont
  {Virosztek}(2015)}]{Petz2015}%
  \BibitemOpen
  \bibfield  {author} {\bibinfo {author} {\bibfnamefont {D.}~\bibnamefont
  {Petz}}\ and\ \bibinfo {author} {\bibfnamefont {D.}~\bibnamefont
  {Virosztek}},\ }\href {\doibase 10.7153/mia-18-41} {\bibfield  {journal}
  {\bibinfo  {journal} {Mathematical Inequalities and Applications}\ }\textbf
  {\bibinfo {volume} {18}},\ \bibinfo {pages} {555} (\bibinfo {year} {2015})},\
  \Eprint {http://arxiv.org/abs/1403.7062} {arXiv:1403.7062} \BibitemShut
  {NoStop}%
\bibitem [{\citenamefont {Appel}\ \emph {et~al.}(2017)\citenamefont {Appel},
  \citenamefont {Huber},\ and\ \citenamefont {Kl{\"{o}}ckl}}]{Appel2017}%
  \BibitemOpen
  \bibfield  {author} {\bibinfo {author} {\bibfnamefont {P.}~\bibnamefont
  {Appel}}, \bibinfo {author} {\bibfnamefont {M.}~\bibnamefont {Huber}}, \ and\
  \bibinfo {author} {\bibfnamefont {C.}~\bibnamefont {Kl{\"{o}}ckl}},\ }\href
  {http://arxiv.org/abs/1710.02473} {\  (\bibinfo {year} {2017})},\ \Eprint
  {http://arxiv.org/abs/1710.02473} {arXiv:1710.02473} \BibitemShut {NoStop}%
\bibitem [{\citenamefont {Eltschka}\ and\ \citenamefont
  {Siewert}(2015)}]{Jens1}%
  \BibitemOpen
  \bibfield  {author} {\bibinfo {author} {\bibfnamefont {C.}~\bibnamefont
  {Eltschka}}\ and\ \bibinfo {author} {\bibfnamefont {J.}~\bibnamefont
  {Siewert}},\ }\href {\doibase 10.1103/PhysRevLett.114.140402} {\bibfield
  {journal} {\bibinfo  {journal} {Phys. Rev. Lett.}\ }\textbf {\bibinfo
  {volume} {114}},\ \bibinfo {pages} {140402} (\bibinfo {year}
  {2015})}\BibitemShut {NoStop}%
\bibitem [{\citenamefont {Eltschka}\ \emph {et~al.}(2018)\citenamefont
  {Eltschka}, \citenamefont {Huber}, \citenamefont {G\"uhne},\ and\
  \citenamefont {Siewert}}]{Jens2}%
  \BibitemOpen
  \bibfield  {author} {\bibinfo {author} {\bibfnamefont {C.}~\bibnamefont
  {Eltschka}}, \bibinfo {author} {\bibfnamefont {F.}~\bibnamefont {Huber}},
  \bibinfo {author} {\bibfnamefont {O.}~\bibnamefont {G\"uhne}}, \ and\
  \bibinfo {author} {\bibfnamefont {J.}~\bibnamefont {Siewert}},\ }\href
  {\doibase 10.1103/PhysRevA.98.052317} {\bibfield  {journal} {\bibinfo
  {journal} {Phys. Rev. A}\ }\textbf {\bibinfo {volume} {98}},\ \bibinfo
  {pages} {052317} (\bibinfo {year} {2018})}\BibitemShut {NoStop}%
\bibitem [{\citenamefont {de~Vicente}\ and\ \citenamefont
  {Huber}(2011)}]{Julio}%
  \BibitemOpen
  \bibfield  {author} {\bibinfo {author} {\bibfnamefont {J.~I.}\ \bibnamefont
  {de~Vicente}}\ and\ \bibinfo {author} {\bibfnamefont {M.}~\bibnamefont
  {Huber}},\ }\href {\doibase 10.1103/PhysRevA.84.062306} {\bibfield  {journal}
  {\bibinfo  {journal} {Phys. Rev. A}\ }\textbf {\bibinfo {volume} {84}},\
  \bibinfo {pages} {062306} (\bibinfo {year} {2011})}\BibitemShut {NoStop}%
\bibitem [{\citenamefont {Laskowski}\ \emph {et~al.}(2011)\citenamefont
  {Laskowski}, \citenamefont {Markiewicz}, \citenamefont {Paterek},\ and\
  \citenamefont {\ifmmode~\dot{Z}\else \.{Z}\fi{}ukowski}}]{Zukowski}%
  \BibitemOpen
  \bibfield  {author} {\bibinfo {author} {\bibfnamefont {W.}~\bibnamefont
  {Laskowski}}, \bibinfo {author} {\bibfnamefont {M.}~\bibnamefont
  {Markiewicz}}, \bibinfo {author} {\bibfnamefont {T.}~\bibnamefont {Paterek}},
  \ and\ \bibinfo {author} {\bibfnamefont {M.}~\bibnamefont
  {\ifmmode~\dot{Z}\else \.{Z}\fi{}ukowski}},\ }\href {\doibase
  10.1103/PhysRevA.84.062305} {\bibfield  {journal} {\bibinfo  {journal} {Phys.
  Rev. A}\ }\textbf {\bibinfo {volume} {84}},\ \bibinfo {pages} {062305}
  (\bibinfo {year} {2011})}\BibitemShut {NoStop}%
\bibitem [{\citenamefont {Kl\"ockl}\ and\ \citenamefont
  {Huber}(2015)}]{Claude1}%
  \BibitemOpen
  \bibfield  {author} {\bibinfo {author} {\bibfnamefont {C.}~\bibnamefont
  {Kl\"ockl}}\ and\ \bibinfo {author} {\bibfnamefont {M.}~\bibnamefont
  {Huber}},\ }\href {\doibase 10.1103/PhysRevA.91.042339} {\bibfield  {journal}
  {\bibinfo  {journal} {Phys. Rev. A}\ }\textbf {\bibinfo {volume} {91}},\
  \bibinfo {pages} {042339} (\bibinfo {year} {2015})}\BibitemShut {NoStop}%
\bibitem [{\citenamefont {Asadian}\ \emph {et~al.}(2016)\citenamefont
  {Asadian}, \citenamefont {Erker}, \citenamefont {Huber},\ and\ \citenamefont
  {Kl\"ockl}}]{Claude2}%
  \BibitemOpen
  \bibfield  {author} {\bibinfo {author} {\bibfnamefont {A.}~\bibnamefont
  {Asadian}}, \bibinfo {author} {\bibfnamefont {P.}~\bibnamefont {Erker}},
  \bibinfo {author} {\bibfnamefont {M.}~\bibnamefont {Huber}}, \ and\ \bibinfo
  {author} {\bibfnamefont {C.}~\bibnamefont {Kl\"ockl}},\ }\href {\doibase
  10.1103/PhysRevA.94.010301} {\bibfield  {journal} {\bibinfo  {journal} {Phys.
  Rev. A}\ }\textbf {\bibinfo {volume} {94}},\ \bibinfo {pages} {010301}
  (\bibinfo {year} {2016})}\BibitemShut {NoStop}%
\bibitem [{\citenamefont {Schwemmer}\ \emph {et~al.}(2015)\citenamefont
  {Schwemmer}, \citenamefont {Knips}, \citenamefont {Tran}, \citenamefont
  {de~Rosier}, \citenamefont {Laskowski}, \citenamefont {Paterek},\ and\
  \citenamefont {Weinfurter}}]{Weinfurter}%
  \BibitemOpen
  \bibfield  {author} {\bibinfo {author} {\bibfnamefont {C.}~\bibnamefont
  {Schwemmer}}, \bibinfo {author} {\bibfnamefont {L.}~\bibnamefont {Knips}},
  \bibinfo {author} {\bibfnamefont {M.~C.}\ \bibnamefont {Tran}}, \bibinfo
  {author} {\bibfnamefont {A.}~\bibnamefont {de~Rosier}}, \bibinfo {author}
  {\bibfnamefont {W.}~\bibnamefont {Laskowski}}, \bibinfo {author}
  {\bibfnamefont {T.}~\bibnamefont {Paterek}}, \ and\ \bibinfo {author}
  {\bibfnamefont {H.}~\bibnamefont {Weinfurter}},\ }\href {\doibase
  10.1103/PhysRevLett.114.180501} {\bibfield  {journal} {\bibinfo  {journal}
  {Phys. Rev. Lett.}\ }\textbf {\bibinfo {volume} {114}},\ \bibinfo {pages}
  {180501} (\bibinfo {year} {2015})}\BibitemShut {NoStop}%
\bibitem [{\citenamefont {Bertlmann}\ and\ \citenamefont
  {Krammer}(2008)}]{Bertlmann2008}%
  \BibitemOpen
  \bibfield  {author} {\bibinfo {author} {\bibfnamefont {R.~A.}\ \bibnamefont
  {Bertlmann}}\ and\ \bibinfo {author} {\bibfnamefont {P.}~\bibnamefont
  {Krammer}},\ }\href {\doibase 10.1088/1751-8113/41/23/235303} {\bibfield
  {journal} {\bibinfo  {journal} {Journal of Physics A: Mathematical and
  Theoretical}\ }\textbf {\bibinfo {volume} {41}},\ \bibinfo {pages} {235303}
  (\bibinfo {year} {2008})},\ \Eprint {http://arxiv.org/abs/0806.1174}
  {arXiv:0806.1174} \BibitemShut {NoStop}%
\end{thebibliography}%

\appendix

\section{Bloch decomposition}\label{AA1}

Recall the Schmidt decomposition:
\begin{theorem}\label{AA1t1}
	Suppose $|\psi\rangle$ is a pure state of a composite system, $AB$. Then there exist orthonormal bases $\big\{|X_i\rangle\big\}_{i=1}^{d_A}$ and $\big\{|Y_i\rangle\big\}_{i=1}^{d_B}$ for the systems $A$ and $B$ respectively, such that
	\begin{equation}
	|\psi\rangle=\sum\limits_{i=1}^{N}c_i|X_i\rangle\otimes|Y_i\rangle,
	\end{equation}
	where $N=min(d_A,d_B)$ and $c_i$ are non-negative real numbers satisfying $\sum\limits_{i=1}^{N}c_i^2= 1$ known as Schmidt co-efficients.
\end{theorem}
A consequence of the Schmidt decomposition is
\begin{cor}\label{AA1c1}
	For every state $\rho_{A}\in\mathcal{H}$ of a quantum system there is a pure state $|\psi\rangle\langle\psi|_{AR}\in\mathcal{H}^*\otimes\mathcal{H}$ such that $\Tr_R(|\psi\rangle\langle\psi|_{AR})=\rho_A$.
\end{cor}

We already know that there exists a orthogonal operator basis for the Hilbert Schmidt space. According to \cite{Bertlmann2008}, every state can be represented in a Bloch basis.

\begin{theorem}\label{AA1t2}
	For the Hilbert space $\mathcal{H}^*\otimes\mathcal{H}$ with $dim(\mathcal{H})=d$ there exists a basis $\big\{X_i\big\}_{i=0}^{d^2-1}$ such that
	
\begin{itemize}
	\item The identity matrix is included, $X_0=\mathds{1}_d$.\\
	\item The other $d^2-1$ elements are traceless Hermitian matrices.\\
	\item The basis elements are mutually orthogonal, $\Tr(X_i^\dagger X_j)=d\ \delta_{ij}$.
\end{itemize}
Every state $\rho$ can be represented in the Bloch picture
\begin{equation}\label{AA1t21}
\rho=\frac{1}{d}\bigg(\mathds{1}_d+\sum\limits_{i=1}^{d^2-1}b_i X_i\bigg)=\frac{1}{d}\bigg(\mathds{1}_d+\vec{\mathbf{b}}\cdot\vec{\mathbf{\Gamma}}\bigg)
\end{equation}
where $\vec{\mathbf{b}}\in\mathbb{R}^{d^2-1}$ is the Bloch vector, $\|\vec{\mathbf{b}}\|^2\le d-1$. The entries are given by $b_i=\langle X_i\rangle=\Tr(\rho X_i)$.
\end{theorem}

It follows immediately that
\begin{align}\label{AA1t22}
\Tr(\rho^2)=&\frac{1}{d^2}\Tr\bigg(\big(\mathds{1}_d+\sum\limits_{i=1}^{d^2-1}b_i X_i\big)^2\bigg)^2\\
=&\frac{1}{d^2}\Tr\bigg(\mathds{1}_d+2\sum\limits_{i=1}^{d^2-1}b_i X_i+2\sum\limits_{i,j=1}^{d^2-1}b_ib_j X_i X_j+\sum\limits_{i=1}^{d^2-1}b_i^2 X_i^2\bigg)\notag \\
=&\frac{1}{d}\bigg(1+\sum\limits_{i=1}^{d^2-1}b_i^2\bigg)\notag \\
=&\frac{1}{d}\bigg(1+\|\vec{\mathbf{b}}\|^2\bigg)\notag
\end{align}

\section{Notes on inhomogeneous subadditivity}
\subsection{Alternative proof of Lemma \ref{ISl1}.}\label{AB1}
\begin{proof}
A direct way to show this result comes from the obvious inequality of the correlation tensor $\|C_{AB}\|_2^2\ge 0$, implying

\begin{align}
S_L(\rho_{AB})=& 1-\frac{1}{d_Ad_B}(1+\|C_A\|_2^2+\|C_B\|_2^2+\|C_{AB}\|_2^2)\\
\le& 1-\frac{1}{d_Ad_B}(1+\|C_A\|_2^2+\|C_B\|_2^2)\\
=& \frac{1}{d_B}S_L(\rho_{A})+\frac{1}{d_A}S_L(\rho_{B})+\frac{(1-d_A)(1-d_B)}{d_Ad_B}
\end{align}
\end{proof}

\subsection{Proof of Proposition \ref{ISp1}.}\label{AB2}
\begin{proof}
The inequality from Theorem \ref{LEt2}. can be rewritten to the equivalent form:

\begin{align}
1-\frac{d_{A}d_{B}}{4}(1-S_{L}(\rho_{AB})+\frac{1}{d_{A}d_{B}})^{2} &\le S_{L}(\rho_{A})+S_{L}(\rho_{B})-S_{L}(\rho_{A})S_{L}(\rho_{B})\\
(1-S_{L}(\rho_{AB})+\frac{1}{d_{A}d_{B}})^{2} &\ge \frac{4}{d_{A}d_{B}}(1-S_{L}(\rho_{A}))(1-S_{L}(\rho_{B})) \label{AB21}\\
1-S_{L}(\rho_{AB})+\frac{1}{d_{A}d_{B}} &\ge 2\sqrt{\frac{1}{d_{A}d_{B}}(1-S_{L}(\rho_{A}))(1-S_{L}(\rho_{B}))}\\
S_{L}(\rho_{AB}) &\le 1+\frac{1}{d_{A}d_{B}}- 2\sqrt{\frac{1}{d_{A}d_{B}}(1-S_{L}(\rho_{A}))(1-S_{L}(\rho_{B}))}
\label{AB22}
\end{align}
In line (\ref{AB21}) the square root can be taken, since both sides are positive.
If we subtract the right hand side of the equation from Lemma \ref{ISl1}. from the right hand side of Equation (\ref{AB22}) we get
\begin{align}
&\bigg(1+\frac{1}{d_{A}d_{B}}-2\sqrt{\frac{1}{d_{A}d_{B}}(1-S_{L}(\rho_{A}))(1-S_{L}(\rho_{B}))}\bigg)\\
&-\bigg(\frac{1}{d_{B}}S_{L}(\rho_{A})+\frac{1}{d_{A}}S_{L}(\rho_{B})+\frac{(d_{A}-1)(d_{B}-1)}{d_{A}d_{B}}\bigg)\notag \\
&=-2\sqrt{\frac{1}{d_{A}d_{B}}(1-S_{L}(\rho_{A}))(1-S_{L}(\rho_{B}))}+\frac{1}{d_{B}}(1-S_{L}(\rho_{A}))+\frac{1}{d_{A}}(1-S_{L}(\rho_{B}))\\
&=\bigg(\sqrt{\frac{1}{d_{B}}(1-S_{L}(\rho_{A}))}-\sqrt{\frac{1}{d_{A}}(1-S_{L}(\rho_{B})})\bigg)^{2} \ge 0
\end{align}
which shows that the inequality in Lemma \ref{ISl1}. is sharper than inequality (\ref{AB22}).
\end{proof}

\section{Proof of Proposition \ref{SBp1}.}\label{AC}

\begin{definition2}
	Let's define the following sets by abuse of notation
	\begin{align}
	ISA&:=\big\{(x,y,z)|\ 0\le x\le D_A,\ 0\le y\le D_B,\ z=h(x,y)\big\}\\
	DSSA&:=\big\{(x,y,z)|\ 0\le x\le D_A,\ 0\le y\le D_B,\ z=g(x,y)\big\}
	\end{align}
	as the graphs of these functions.
	Let 
	\begin{align}
	\Omega&:=\big\{(x,y)|\ 0\le x\le D_A,\ 0\le y\le D_B,\ y=D_A\big(\frac{x}{D_B}-1+2\sqrt{1-\frac{x}{D_B}}\big)\big\}
	\end{align}
	denote the set where equality holds for inequality (\ref{DSSAt12}) and 
	\begin{align}
	\overline{\Omega}&:=\big\{(x,y,z)|\ 0\le x\le D_A,\ 0\le y\le D_B,\ 0\le z\le D_{AB},\ y=D_A\big(\frac{x}{D_B}-1+2\sqrt{1-\frac{x}{D_B}}\big)\big\}
	\end{align}
	the extension of $\Omega$ along the z-axis.
\end{definition2}

\begin{claim}
	The intersection between the sets $DSSA$ and $\overline{\Omega}$ is exactly the intersection between the sets $ISA$ and $\overline{\Omega}$. Let us denote this intersection by
	\begin{align}
	\Gamma:=DSSA\cap\overline{\Omega}=ISA\cap\overline{\Omega}
	\end{align}
\end{claim}

\begin{proof}
	Since $\Omega$ divides the domain $\big\{(x,y)|\ 0\le x\le D_A,\ 0\le y\le D_B\big\} $ of $g(x,y)$ and $h(x,y)$ in two parts, obviously $DSSA\cap\overline{\Omega}$ and $ISA\cap\overline{\Omega}$ are non-empty.\\
	Since $\Gamma\subset\overline{\Omega}$, all we have to do is to check that the functions $g(x,y)$ and $h(x,y)$ coincide on $\Omega$.
	\begin{align*}
	& g(x,y)-h(x,y)\\
	&= x+y-\frac{2(d_A-1)(d_B-1)}{d_A d_B}\big(1-\sqrt{1-\frac{d_A}{d_A-1} x}\big)\big(1-\sqrt{1-\frac{d_B}{d_B-1} y}\big)\\
	& -\frac{1}{d_{B}}x-\frac{1}{d_{A}}y-\frac{(d_{A}-1)(d_{B}-1)}{d_{A}d_{B}}\\
	&=\frac{d_B-1}{d_B}x+\frac{d_A-1}{d_A}y-\frac{(d_{A}-1)(d_{B}-1)}{d_{A}d_{B}}\\
	&-\frac{(d_A-1)(d_B-1)}{2d_Ad_B}\big(1-\frac{d_A}{d_A-1}x+\frac{d_B}{d_B-1}y\big)\big(1+\frac{d_A}{d_A-1}x-\frac{d_B}{d_B-1}y\big)\\
	&=\frac{d_B-1}{d_B}x+\frac{d_A-1}{d_A}y-\frac{3(d_A-1)(d_B-1)}{2d_A d_B}+\frac{(d_{A}-1)(d_{B}-1)}{2d_{A}d_{B}}\big(\frac{d_A}{d_A-1}x-\frac{d_B}{d_B-1}y\big)^2\\
	&=\frac{2(d_{A}-1)}{d_{A}}y-\frac{5(d_{A}-1)(d_{B}-1)}{2d_{A}d_{B}}+\frac{2(d_{A}-1)(d_{B}-1)}{d_{A}d_{B}}\sqrt{1-\frac{d_B}{d_B-1}\ y}\\
	&+\frac{(d_{A}-1)(d_{B}-1)}{2d_{A}d_{B}}\big(2\sqrt{1-\frac{d_B}{d_B-1}\ y}-1\big)^2\\
	&=0
	\end{align*}
\end{proof}

We conclude that $g(x,y)$ and $h(x,y)$ can be combined to a continuous surface, if we put $DSSA$ (left of $\overline{\Omega}$) and $ISA$ (right of $\overline{\Omega}$) together.\\
Now we only have to check the differentiability.\\
Since $g(x,y)$ and $h(x,y)$ are continuously differentiable functions, it suffices to show that the transition is too.\\
Look at the gradient of both functions on the overlap $\Omega$.

\begin{align}
    \nabla g(x,y)=\begin{pmatrix}
	1-D_B\big(1-\sqrt{1-\frac{y}{D_B}}\big)\frac{1}{\sqrt{1-\frac{x}{D_A}}} \\
	\\
	1-D_A(1-\sqrt{1-\frac{x}{D_A}})\frac{1}{\sqrt{1-\frac{y}{D_B}}} 
	\end{pmatrix} \label{AC1}
	\end{align}
	
	\begin{align}
	\nabla g(x,y)|_{\Omega}=\begin{pmatrix}
	\frac{1}{d_B}\\ 
	\\
	\frac{1}{d_A}
	\end{pmatrix}=\nabla h(x,y)  \label{AC2}
\end{align}

This shows that $f(x,y)$ is continuously differentiable.

\section{Proof of Lemma \ref{SBl1}.}\label{AD}

We want to show that the inequality
\begin{align}
S_L(\rho_{AB})\le f(S_L(\rho_{A}),S_L(\rho_{B}))
\end{align}
is the sharpest possible inequality for a general bipartite quantum state. This can be shown by simply constructing states for which equality holds. The entropy vectors of those states lie on the boundary of the entropy body.\\

\begin{lemma2}
Let $\rho_{AB}(\alpha)$ be a family of states of the form
\[\rho_{AB}(\alpha):=\alpha \mu_{AB}+(1-\alpha)\frac{\mathds{1}_{d_{AB}}}{d_{AB}}\]
where $\mu_{AB}$ is an arbitrary state, $\mathds{1}_{d_{AB}}$ the identity operator and $0 \le \alpha \le 1$.\\
Then their entropies $\big(S_{L}(\rho_{A}(\alpha)),S_{L}(\rho_{B}(\alpha)),S_{L}(\rho_{AB}(\alpha))\big)$ lie on a line in $\mathbb{R}^3$.
\end{lemma2}

\begin{proof}
We have
{\allowdisplaybreaks
\begin{align}
S_{L}(\rho_{AB}(\alpha))&=1-\Tr(\rho_{AB}(\alpha)^2)\\
&=1-\Tr((\alpha \mu_{AB}+(1-\alpha)\frac{\mathds{1}_{d_{AB}}}{d_{AB}})^2)\notag \\
&=1-\alpha^2\Tr( \mu_{AB}^2)-2\alpha(1-\alpha)\frac{\Tr(\mu_{AB})}{d_{AB}}-\frac{(1-\alpha)^2}{d_{AB}}\Tr(\frac{\mathds{1}_{d_{AB}}}{d_{AB}})\notag \\
&=\big(\alpha^2-\alpha^2\Tr( \mu_{AB}^2)\big)+(1-\alpha^2)-\frac{2\alpha}{d_{AB}}(1-\alpha)-\frac{(1-\alpha)^2}{d_{AB}}\notag \\
&=\big(\alpha^2-\alpha^2\Tr( \mu_{AB}^2)\big)+\big((1-\alpha^2)\frac{d_{AB}-1}{d_{AB}}\big)\notag \\
&=\alpha^2S_{L}( \mu_{AB})+(1-\alpha^2)S_{L}(\frac{\mathds{1}_{d_{AB}}}{d_{AB}})\notag 
\end{align}}
Using
\begin{align}
\Tr_{B}(\rho_{AB}(\alpha))&=\sum_{m} \langle m |_{B}\rho_{AB}(\alpha)|m \rangle_{B}\\
&=\alpha \sum_{m} \langle m |_{B}\mu_{AB}|m \rangle_{B}+(1-\alpha)\sum_{m} \langle m |_{B}\frac{\mathds{1}_{d_{AB}}}{d_{AB}}|m \rangle_{B}\notag \\
&=\alpha \mu_{A}+(1-\alpha)\frac{\mathds{1}_{d_{A}}}{d_{A}}\notag 
\end{align}
we can write
\begin{align}
S_{L}(\rho_{A}(\alpha))&=1-\Tr(\rho_{A}(\alpha)^2)\\
&=1-\Tr((\Tr_{B}(\rho_{AB}(\alpha)))^2)\notag \\
&=1-\Tr((\Tr_{B}(\alpha \mu_{A}+(1-\alpha)\frac{\mathds{1}_{d_{A}}}{d_{A}}))^2)\notag \\
&=\alpha^2S_{L}( \mu_{A})+(1-\alpha^2)S_{L}(\frac{\mathds{1}_{d_{A}}}{d_{A}})\notag 
\end{align}
and analogously
\begin{align}
S_{L}(\rho_{B}(\alpha))=\alpha^2S_{L}( \mu_{B})+(1-\alpha^2)S_{L}(\frac{\mathds{1}_{d_{B}}}{d_{B}})
\end{align}
Let
\begin{align}
\vec{\mathbf{v}}(\rho_{AB}):=(S_{L}(\rho_{A}),S_{L}(\rho_{B}),S_{L}(\rho_{AB}))
\end{align}
be the entropy vector as defined in Eq. (\ref{VN6}).\\
We then have
\begin{align}
\vec{\mathbf{v}}(\mu_{AB})=(S_{L}(\mu_{A}),S_{L}(\mu_{B}),S_{L}(\mu_{AB}))\\
\notag \\
\vec{\mathbf{v}}(\frac{\mathds{1}_{d}}{d})=(S_{L}(\frac{\mathds{1}_{d_{A}}}{d_{A}}),S_{L}(\frac{\mathds{1}_{d_{B}}}{d_{B}}),S_{L}(\frac{\mathds{1}_{d_{AB}}}{d_{AB}}))
\end{align}
and we can write
\begin{align}
\vec{\mathbf{v}}(\rho_{AB}(\alpha))=\vec{\mathbf{v}}(\frac{\mathds{1}_{d_{AB}}}{d_{AB}})+\alpha^2(\vec{\mathbf{v}}(\mu_{AB})-\vec{\mathbf{v}}(\frac{\mathds{1}_{d_{AB}}}{d_{AB}}))
\end{align}

We therefore know that the entropy vector of any convex combination of an arbitrary state and the completely mixed state $\rho_{AB}(\alpha):=\alpha \mu_{AB}+(1-\alpha)\frac{\mathds{1}_{d_{AB}}}{d_{AB}}$ lies on the line between the entropy vector of the state and the entropy vector of the completely mixed state (upper right corner).
\end{proof}

Note however, that this is not true for two arbitrary states.\\

The set $\Gamma$ can be described as a parametrised curve in the following way

\begin{align}
\Gamma=\bigg\{\big(\gamma_1(t),\gamma_2(t),\gamma_3(t)\big)|\ 0\le t\le D_A\bigg\}
\end{align}
where
\begin{align}
\gamma(t)=\begin{pmatrix}
{\gamma_1(t)}\\
{\gamma_2(t)}\\
{\gamma_3(t)}
\end{pmatrix}:=
\begin{pmatrix}
{t}\\
{\frac{d_B-1}{d_B}\big(\frac{d_A}{d_A-1}\ t-1+2\sqrt{1-\frac{d_A}{d_A-1}t}\big)}\\
{\frac{d_A+d_B-2}{d_{B}(d_A-1)}\ t+\frac{2(d_B-1)}{d_B d_A}\sqrt{1-\frac{d_A}{d_A-1}t}+\frac{(d_{A}-2)(d_{B}-1)}{d_{A}d_{B}}}
\end{pmatrix}
\end{align}

This can be seen by simply substituting $x=t$. The second argument $\gamma_2(t)$ follows immediately from the definition of $\Gamma$. To obtain the third entry $\gamma_3(t)$, plug the first two into the definition of  $ISA$.\\

\begin{lemma2}
The set of states depending on $0\le\alpha\le1$.
\begin{align}
\rho_{AB}(\alpha)=\alpha\mu_{1}+(1-\alpha)\mu_{2}
\end{align}
where
\begin{align*}
\mu_{1}=|0\rangle\langle0|\otimes\frac{\mathds{1}_{d_B}}{d_B}\qquad\mu_{2}=\frac{\mathds{1}_{d_A}}{d_A}\otimes|0\rangle\langle0|
\end{align*}
form a family of states with entropies on $\Gamma$ and on $ISA$.
\end{lemma2}

\begin{proof}
We can calculate
{\allowdisplaybreaks
\begin{align}
S_{L}(\rho_{AB}(\alpha)) &=1-\Tr(\rho_{AB}^2(\alpha))\\
&=1-\Tr((\alpha\mu_{1}+(1-\alpha)\mu_{2})^2)\notag \\
&=1-\alpha^2\Tr(\mu_{1}^2)-2\alpha(1-\alpha)\Tr(\mu_{1}\mu_{2})-(1-\alpha)^2\Tr(\mu_{2}^2)\notag \\
&=1-\frac{\alpha^2}{d_B}-\frac{2\alpha(1-\alpha)}{d_A d_B}-\frac{(1-\alpha)^2}{d_A}\notag \\
&=\frac{d_A-1}{d_A}+\frac{2(d_B-1)}{d_A d_B}\alpha-\frac{d_A+d_B-2}{d_A d_B}\alpha^2\notag \\
\notag \\
S_{L}(\rho_{A}(\alpha)) &=1-\Tr(\rho_{A}^2(\alpha))\\
&=1-\Tr((\Tr_{B}(\rho_{AB}(\alpha)))^2)\notag \\
&=1-\Tr((\alpha |0\rangle\langle0|+\frac{1-\alpha}{d_A}\mathds{1}_{d_A})^2)\notag \\
&=1-\alpha^2 -\frac{2\alpha(1-\alpha)}{d_A}-\frac{(1-\alpha)^2}{d_A}\notag \\
&=\frac{d_A-1}{d_A}-\frac{d_A-1}{d_A}\alpha^2\notag \\
\notag \\
S_{L}(\rho_{B}(\alpha)) &=1-\Tr(\rho_{B}^2(\alpha))\\
&=1-\Tr((\Tr_{A}(\rho_{AB}(\alpha)))^2)\notag \\
&=1-\Tr(((1-\alpha) |0\rangle\langle0|+\frac{\alpha}{d_B}\mathds{1}_{d_B})^2)\notag \\
&=1-\frac{\alpha^2}{d_B}-\frac{2\alpha(1-\alpha)}{d_B}-(1-\alpha)^2\notag \\
&=2\frac{d_B-1}{d_B}\alpha-\frac{d_B-1}{d_B}\alpha^2\notag 
\end{align}}
\hspace{\fill}

Using $t(\alpha)=\frac{d_A-1}{d_A}-\frac{d_A-1}{d_A}\alpha^2$ we can now write

\begin{align}
\gamma(t(\alpha))=&
\begin{pmatrix}
{t(\alpha)}\\
\\
{\frac{d_B-1}{d_B}\big(\frac{d_A}{d_A-1}t(\alpha)-1+2\sqrt{1-\frac{d_A}{d_A-1}t(\alpha)}\big)}\\
\\
{\frac{1}{d_{B}}t(\alpha)+\frac{d_B-1}{d_B d_A}\big(\frac{d_A}{d_A-1}t(\alpha)-1+2\sqrt{1-\frac{d_A}{d_A-1}t(\alpha)}\big)+\frac{(d_{A}-1)(d_{B}-1)}{d_{A}d_{B}}}
\end{pmatrix}\\
\notag \\
=&\begin{pmatrix}
{\frac{d_A-1}{d_A}-\frac{d_A-1}{d_A}\alpha^2}\\
\\
{2\frac{d_B-1}{d_B}\alpha-\frac{d_B-1}{d_B}\alpha^2}\\
\\
{\frac{d_A-1}{d_A}+\frac{2(d_B-1)}{d_A d_B}\alpha-\frac{d_A+d_B-2}{d_A d_B}\alpha^2}
\end{pmatrix}=
\begin{pmatrix}
{S_{L}(\rho_{A}(\alpha))}\\
{S_{L}(\rho_{B}(\alpha))}\\
{S_{L}(\rho_{AB}(\alpha))}
\end{pmatrix}\notag
\end{align}
\hspace{\fill}

Therefore, we have found states whose entropies lie on $\Gamma$. We know, that the entropies of the mixture of an arbitrary state and the completely mixed state lie on a line. Since $\Gamma$ lies on $ISA$, we can conclude that $ISA$ (since it is linear) is the sharpest possible inequality for states on the right hand side of $\overline{\Omega}$.
\end{proof}

\begin{lemma2}
The set of states depending on $0\le\alpha\le1$ and $0\le\beta\le1$ with $\alpha+\beta \le 1$
\begin{align}
\tilde{\rho}_{AB}(\alpha,\beta)=\alpha\mu_{1}+\beta\mu_{2}+(1-\alpha-\beta)\mu_{3}
\end{align}

where $\mu_{1}$ and $\mu_{2}$ are defined as before and $\mu_{3}=|00\rangle\langle00|$, %.\\
% \begin{align}
%\rho_{AB}(\alpha)=\alpha\mu_{1}+(1-\alpha)\mu_{2}
%\end{align}
%where
%\begin{align*}
%\mu_{1}=|0\rangle\langle0|\otimes\frac{\mathds{1}_{d_B}}{d_B}\qquad\mu_{2}=\frac{\mathds{1}_{d_%A}}{d_A}\otimes|0\rangle\langle0|
%\end{align*}
form a family of states with entropies on $DSSA$.
\end{lemma2}

\begin{proof}
In an analogous way to the previous proof we can calculate

{\allowdisplaybreaks
\begin{align}
S_{L}(\tilde{\rho}_{AB}(\alpha,\beta)) =& 1-\Tr(\tilde{\rho}_{AB}^2(\alpha,\beta))\\
=& 1-\Tr((\alpha\mu_{1}+\beta\mu_{2}+(1-\alpha-\beta)\mu_{3})^2)\notag \\
=& 1-\alpha^2\Tr(\mu_{1}^2)-\beta^2\Tr(\mu_{2}^2)-(1-\alpha-\beta)^2\Tr(\mu_{3}^2)\notag \\
&-2\alpha\beta \Tr(\mu_{1}\mu_{2})-2\alpha(1-\alpha-\beta)\Tr(\mu_{1}\mu_{3})-2\beta(1-\alpha-\beta)\Tr(\mu_{2}\mu_{3})\notag \\
=& 1-\frac{\alpha^2}{d_B}-\frac{\beta^2}{d_A}-(1-\alpha-\beta)^2\notag \\
&-\frac{2\alpha\beta}{d_Ad_B} -\frac{2\alpha(1-\alpha-\beta)}{d_B}-\frac{2\beta(1-\alpha-\beta)}{d_A}\notag \\
=& \frac{d_B-1}{d_B}(2\alpha-\alpha^2)+\frac{d_A-1}{d_A}(2\beta-\beta^2)-\frac{2(d_A-1)(d_B-1)}{d_Ad_B}\alpha\beta\notag \\
=& D_B(2\alpha-\alpha^2)+D_A(2\beta-\beta^2)-2D_AD_B\alpha\beta\notag \\
\notag \\
S_{L}(\tilde{\rho}_{A}(\alpha,\beta)) =&1-\Tr(\tilde{\rho}_{A}^2(\alpha,\beta))\\
=&1-\Tr((\Tr_{B}(\tilde{\rho}_{AB}(\alpha))^2)\notag \\
=&1-\Tr((\frac{\beta}{d_A}\mathds{1}_{d_A}+(1-\beta)|0\rangle\langle0|)^2)\notag \\
=&1-\frac{\beta^2}{d_A}-\frac{2\beta(1-\beta)}{d_A}-(1-\beta)^2\notag \\
=&\frac{d_A-1}{d_A}(2\beta-\beta^2)\notag \\
=&D_A(2\beta-\beta^2)\notag \\
\notag \\
S_{L}(\tilde{\rho}_{B}(\alpha,\beta)) =&1-\Tr(\tilde{\rho}_{B}^2(\alpha,\beta))\\
=&1-\Tr((\Tr_{A}(\tilde{\rho}_{AB}(\alpha))^2)\notag \\
=&1-\Tr((\frac{\alpha}{d_B}\mathds{1}_{d_B}+(1-\alpha)|0\rangle\langle0|)^2)\notag \\
=&1-\frac{\alpha^2}{d_B}-\frac{2\alpha(1-\alpha)}{d_B}-(1-\alpha)^2\notag \\
=&\frac{d_B-1}{d_B}(2\alpha-\alpha^2)\notag \\
=&D_B(2\alpha-\alpha^2)\notag 
\end{align}}
\hspace{\fill}

The surface $DSSA$ is given by the equation
\begin{align}
S_{L}(\rho_{AB})=&
\ S_{L}(\rho_{A})+ S_{L}(\rho_{B})\\
& -2D_AD_B\big(1-\sqrt{1-\frac{S_{L}(\rho_{A})}{D_A}} \big)\big(1-\sqrt{1-\frac{S_{L}(\rho_{B})}{D_B} }\big)\notag 
\end{align}
\hspace{\fill}

Since
\begin{align}
& S_{L}(\rho_{A})+ S_{L}(\rho_{B})-2D_AD_B\big(1-\sqrt{1-\frac{S_{L}(\rho_{A})}{D_A}} \big)\big(1-\sqrt{1-\frac{S_{L}(\rho_{B})}{D_B} }\big)\\
=& D_A(2\beta-\beta^2)+D_B(2\alpha-\alpha^2) -2D_AD_B\big(1-\sqrt{1-2\beta+\beta^2}\big)\big(1-\sqrt{1-2\alpha+\alpha^2}\big)\notag \\
=& D_B(2\alpha-\alpha^2)+D_A(2\beta-\beta^2)-2D_AD_B\alpha\beta\notag \\
=& S_{L}(\rho_{AB})\notag 
\end{align}

we have found a set of states whose entropies lie on $DSSA$, proving $DSSA$ to be sharp.
\end{proof}

\section{Proof of Proposition \ref{RPp1}.}\label{AE}

\begin{proof}
Use Lemma \ref{ISl1}. and Eq. (\ref{462}) to conclude that the inequality
	\begin{align}
	S^2(\rho_{AB})&\le-\log_2\big(\frac{2^{-S^2(\rho_A)}}{d_B}\ +\frac{2^{-S^2(\rho_B)}}{d_A}-\frac{1}{d_A d_B}\big) \label{l4611}
	\end{align}
	holds and analogously Theorem \ref{DSSAt1}. to conclude that the inequality
	\begin{align}
	S^2(\rho_{AB})&\le-\log_2\bigg(2^{-S^2(\rho_{A})}+2^{-S^2(\rho_{B})}-1-\label{l4621}\\
	& \qquad2\frac{d_A-1}{d_A}\frac{d_B-1}{d_B}\big(1-\sqrt{\frac{d_A\ 2^{-S^2(\rho_A)}-1}{d_A-1}}\big)\big(1-\sqrt{\frac{d_B\ 2^{-S^2(\rho_B)}-1}{d_B-1}}\big)\bigg)\notag 
	\end{align}
	holds under the assumption that
	\begin{align}
	S^2(\rho_{A})&\le-\log_2\bigg(\frac{d_A-1}{d_A}\bigg(\frac{d_B}{d_B-1}\  2^{-S^2(\rho_{B})}-\frac{1}{d_B-1}-2\sqrt{\frac{d_B\ 2^{-S^2(\rho_B)}-1}{d_B-1}}\bigg)-1\bigg)\\
\end{align}

Since $f(x,y)$, as defined in Proposition \ref{SBp1}., is continously differentiable and $e(x):=1-2^{-x}$ is a diffeomorphism for $x\in\mathbb{R}^+$, also
\begin{equation}
	f_R(x,y)=e^{-1}\big(f\big(e(x),e(y)\big)\big)
\end{equation}
is continuously differentiable.
\end{proof}

\section{Implicit function theorem and inverted function theorem}\label{AF}

The implicit function theorem states:

\begin{theorem}
	Let $U$ be an open set in $\mathbb{R}^{k+n}$ and $F:U\rightarrow\mathbb{R}^{n}$ be a continuously differentiable function. Let $\mathbf{v}\in\mathbb{R}^{k}$ and $\mathbf{w}\in\mathbb{R}^{n}$, such that the function can be written as $F(\mathbf{v},\mathbf{w})$.\\
	Suppose $(\mathbf{v}_0,\mathbf{w}_0)\in U$ such that $F(\mathbf{v}_0,\mathbf{w}_0)=0$ and the Jacobian restricted to $\mathbf{w}$ is invertible, that is
	\begin{equation}
	\det(D\ F(\mathbf{v}_0,\mathbf{w}))\big|_{\mathbf{w}=\mathbf{w}_0}=\begin{vmatrix}
	\frac{\partial F_1}{\partial w_1}(\mathbf{v}_0,\mathbf{w}_0) & ... &\frac{\partial F_1}{\partial w_n}(\mathbf{v}_0,\mathbf{w}_0)\\
	\vdots & \ddots & \vdots\\
	\frac{\partial F_n}{\partial w_1} (\mathbf{v}_0,\mathbf{w}_0)& ... &\frac{\partial F_n}{\partial w_n}(\mathbf{v}_0,\mathbf{w}_0)\\
	\end{vmatrix} \ne 0
	\end{equation}
	
	Then there exists a neighbourhood $V\subset\mathbb{R}^{k}$ of $\mathbf{v}_0$ and a unique function $G:V\rightarrow\mathbb{R}^{n}$, such that $G(\mathbf{v}_0)=\mathbf{w}_0$ and $F(\mathbf{v},G(\mathbf{v}))=0$ for every $\mathbf{v}\in V$.
\end{theorem}

The inverse function theorem states:

\begin{theorem}
	Let $U$ be an open set in $\mathbb{R}^{n}$ and $F:U\rightarrow\mathbb{R}^{n}$ be a continuously differentiable function. Let $\mathbf{v}_0\in\mathbb{R}^{n}$ such that the Jacobian is invertible.\\
	Then there exists a neighbourhood $V\subset\mathbb{R}^{n}$ of $\mathbf{v}_0$, such that $F$ is invertible and the inverse function $F^{-1}$ is continuously differentiable.\\
\end{theorem}

\end{document}